\documentclass{article} %
\usepackage{iclr2027_conference,times}

\usepackage{amsmath,amsfonts,bm, amsthm}
\usepackage{hyperref}
\usepackage{cleveref}
\usepackage{algorithm}
\usepackage{algpseudocode}
\usepackage{url}
\usepackage{nicefrac}
\usepackage{pgfplots}
\usepackage{enumitem}
\usepgfplotslibrary{fillbetween}
\pgfplotsset{compat=1.18}

\definecolor{c1}{HTML}{377EB8} %
\definecolor{c2}{HTML}{FF7F00} %
\definecolor{c3}{HTML}{4DAF4A} %
\definecolor{c4}{HTML}{E41A1C} %
\definecolor{c5}{HTML}{984EA3} %

\newcommand{\addlinewithband}[3]{
    \addplot [name path=upper_#2, draw=none, forget plot]
        table [x index=0, y index=3, col sep=comma] {#1};
    \addplot [name path=lower_#2, draw=none, forget plot]
        table [x index=0, y index=2, col sep=comma] {#1};
    \addplot [#2, fill=#2, fill opacity=0.15, draw=none, forget plot]
        fill between [of=upper_#2 and lower_#2];
    \addplot [#2, line width=1.1pt]
        table [x index=0, y index=1, col sep=comma] {#1};
    \addlegendentry{#3}
}

\def\eqref#1{equation~\ref{#1}}

\def\1{\bm{1}}

\DeclareMathAlphabet{\mathsfit}{\encodingdefault}{\sfdefault}{m}{sl}
\SetMathAlphabet{\mathsfit}{bold}{\encodingdefault}{\sfdefault}{bx}{n}

\DeclareMathOperator*{\argmax}{arg\,max}
\DeclareMathOperator*{\argmin}{arg\,min}

\newcommand{\real}{\mathbb{R}}
\newcommand{\prob}{\mathcal{P}}
\newcommand{\state}{\bm{s}}
\newcommand{\statespace}{\mathcal{S}}

\newcommand{\minaction}{\bm{u}}
\newcommand{\maxaction}{\bm{v}}
\newcommand{\actionspacemin}{\mathcal{U}}
\newcommand{\actionspacemax}{\mathcal{V}}

\newcommand{\dimstate}{{d_{\state}}}
\newcommand{\dimminaction}{{d_{\minaction}}}
\newcommand{\dimmaxaction}{{d_{\maxaction}}}
\newcommand{\xconc}{\bm{x}}
\newcommand{\xspace}{\mathcal{X}}

\newcommand{\transitionkernel}{P}
\newcommand{\df}{\gamma}
\newcommand{\rwdfunc}{r}
\newcommand{\rmax}{R_{\mathrm{max}}}
\newcommand{\return}{\mathcal{R}}

\newcommand{\qfunc}{Q}
\newcommand{\qmax}{\qfunc_{\mathrm{max}}}
\newcommand{\qfunck}{\qfunc_k}
\newcommand{\qfunckone}{Q_{k+1}}
\newcommand{\params}{\theta}

\newcommand{\lipschitz}{L_\mathcal{F}}

\newcommand{\valfunc}{V}

\newcommand{\coveringnumber}{N}
\newcommand{\condidens}{p}

\newcommand{\numiters}{K}
\newcommand{\samplingdist}{\sigma}

\newcommand{\oracle}{\mathcal{O}}

\newcommand{\bellman}{\mathcal{T}}
\newcommand{\approxbellman}{\tilde{\bellman}}

\newcommand{\ynoise}{\xi}

\newcommand{\error}{\varepsilon}
\newtheorem{assumption}{Assumption}
\newtheorem{theorem}{Theorem}
\newtheorem{lemma}{Lemma}
\newtheorem{corollary}{Corollary}
\newtheorem{remark}{Remark}

\newtheorem{fact}{Fact}
\crefname{assumption}{Assumption}{Assumptions}
\Crefname{assumption}{Assumption}{Assumptions}
\crefname{corollary}{Corollary}{Corollaries}
\Crefname{corollary}{Corollary}{Corollaries}
\newcommand{\bb}[1]{\mathbb{#1}}
\newcommand{\ee}{\bb{E}}
\newcommand{\variance}{\mathrm{Var}}
\newcommand{\bestfk}{\bar{f}_k}
\newcommand{\delapprox}{\delta_{\mathrm{approx}}}
\newcommand{\volume}{V}

\newcommand{\dk}{D_k}
\newcommand{\dkone}{D_{k+1}}

\newcommand{\st}[1]{\state_{(#1)}}
\newcommand{\ut}[1]{\minaction_{(#1)}}
\newcommand{\vt}[1]{\maxaction_{(#1)}}

\newcommand{\x}{\bm x}
\newcommand{\y}{\bm y}

\newcommand{\softmax}{\varphi}

\title{$\qfunc$ Can Play That Game:\\Online Fitted $\qfunc$-Iteration for Continuous-Action Zero-Sum Markov Games with Convex-Concave Function Approximation}

\author{Kushagra Gupta$^{1}$, Jingqi Li$^1$, Cade Armstrong$^1$, Lasse Peters$^2$, Ross E. Allen$^3$,\\
\textbf{Ufuk Topcu$^1$, David Fridovich-Keil$^1$}\\
$^1$ The University of Texas at Austin\\
$^2$ The University of California, Berkeley\\
$^3$ Massachusetts Institute of Technology (MIT) Lincoln Labs \\
\texttt{\{kushagrag,\,jingqi.li,\,cadearmstrong,\,utopcu,\,dfk\}@utexas.edu}\\
\texttt{lasse.peters@berkeley.edu, ross.allen@ll.mit.edu}
}

\iclrfinalcopy %
\begin{document}

\maketitle

\begin{abstract}
Zero-sum Markov games arise in a wide variety of sequential  decision-making problems such as adversarial learning and planning against modeled uncertainties. However, prior work on finite-sample guarantees on the learned state-action value function ($\qfunc$-function) for zero-sum Markov games is largely  restricted to finite-action settings, or to continuous-action games in which agents have linear dynamics and quadratic rewards (i.e. linear-quadratic, or LQ, games). A central challenge in extending such guarantees to more general continuous-action zero-sum Markov games is that the associated Bellman operator involves a minimax problem that need not admit a tractable saddle-point solution. To this end, we first introduce a class of neural network function approximators for the $\qfunc$-function that is convex-concave in the players' actions, which guarantees that this minimax problem admits a pure-strategy saddle point. We then study an online variant of fitted $\qfunc$-iteration employing this function class and establish, to the best of our knowledge, the first finite-sample guarantees for non-LQ zero-sum Markov games with continuous states and actions. Our code can be found at \url{https://github.com/CLeARoboticsLab/QCanPlayThatGame}. \\

\end{abstract}

\section{Introduction}
We consider the problem of learning state-action value functions ($\qfunc$-functions) in discounted, infinite-horizon, two-player zero-sum Markov games with \emph{continuous} state-action spaces.
Such games naturally arise in competitive scenarios and robust decision-making, where environmental uncertainties and disturbances are adversarially modeled \citep{zhou1996robust,morimoto2000robust,pinto2017robust}. Prior work has established finite-sample guarantees for zero-sum Markov games in the finite-action setting \citep{lagoudakis2002, perolat2015approximate, zhao2022provably, fan2020theoretical, zhang2021finite, jia2019feature, sidford2020solving, xie2020learning, alacaoglu2022natural,chen2024two,feng2024improving}. However, extending these results to the continuous state-action setting remains challenging. The difficulty stems from the \emph{minimax} problem on $\qfunc$ inside the Bellman operator. At every state, each player wishes to play a security strategy, meaning that one player's policy solves a $\min\max$ problem and the other's solves a $\max\min$.
The solutions to these two distinct problems are mutually consistent at \emph{saddle points}, where both players have non-exploitable security strategies. Saddle points are guaranteed to exist when policies map from state to \emph{mixed strategies} over the action space.
In finite-action spaces, the resulting saddle-point problem can be solved via a linear program; in continuous action spaces however, it is generally intractable to compute mixed-strategy saddle points as mixed-strategy spaces are infinite-dimensional.
Moreover, when policies are restricted to output \emph{pure strategies} at each state, i.e. raw actions, saddle-point solutions generally do not exist for arbitrary $\qfunc$.

In this work, we propose a new class of neural network function approximators of the $\qfunc$-function which is convex in the minimizing player's action and concave in the maximizing player's action. This convex-concave structure addresses the challenge above, as the minimax theorem guarantees that a pure-strategy saddle point exists for the minimax problem inside the Bellman operator. Further, this pure-strategy saddle point can also be reliably computed up to any desired accuracy using existing minimax optimization methods \citep{nemirovski2004,lin2020near}. We analyze an online variant of minimax fitted $\qfunc$-iteration which employs this convex-concave function class and establish finite-sample bounds between the learned $\qfunc$-function and the fixed point of the Bellman operator, assuming that the fixed point admits pure-strategy saddle points. 

The prior work closest to ours establishes finite-sample bounds for online minimax fitted $\qfunc$-iteration with neural network function approximation, but is restricted to finite-action spaces \citep{fan2020theoretical}. In the continuous-action setting, single-player reinforcement learning literature has developed finite-sample analysis for fitted $\qfunc$-iteration with function approximation \citep{antos2007fitted}. However, extending these results to the zero-sum Markov game setting remains underexplored. Within the game setting we consider, theoretical finite-sample results for continuous state-action domains remain largely confined to the linear-quadratic (LQ) setting, where agents have linear transition dynamics and quadratic objectives \citep{bacsar1998dynamic,al2007adaptive,al2007model,zhang2019policy,bu2019global,zhang2021provably}. These results exploit the fact that in the LQ setting, the optimal $\qfunc$-function and the resulting policies admit a closed-form solution via the algebraic Riccati equations \citep{bacsar1998dynamic}. Unfortunately, this Riccati structure does not extend beyond the LQ setting.
In this work we move beyond the LQ setting, and make the following contributions:
\begin{itemize}
    \item \textbf{Contribution 1.} We propose a new class of neural network function approximators for the $\qfunc$-function of zero-sum Markov games. Building on the partially input convex neural network architecture \citep{amos2017input}, this class is convex in the minimizing player's actions and concave in the maximizing player's actions by construction. The convex-concave structure ensures the existence of a saddle-point solution to the minimax problem in the Bellman operator.
    \item \textbf{Contribution 2.} We analyze batched minimax $\qfunc$-learning employing the above function approximation class, and establish finite-sample bounds comparing the learned $\qfunc$-function to the fixed point of the Bellman operator, assuming the fixed point admits pure-strategy saddle points. To the best of our knowledge, this is the first such bound for neural network function approximation in zero-sum Markov games with continuous state-action spaces beyond the LQ setting.
\end{itemize}
Additionally, we empirically investigate whether these convex-concave neural networks lead to stable training when used on a practical version of modified fitted $\qfunc$-iteration. 

\section{Preliminaries}\label{section: prelims}
\paragraph{Notation.} $\bb{Z}_{\geq0}$ denotes the set of non-negative integers, and for $n\in\bb{Z}_{\geq0}$, $[n]$ denotes the set $\{1,2,\cdots,n\}$. $\prob(\statespace)$ denotes the set of all probability measures over a measurable space $\statespace$. $M(\statespace)$ denotes the space of bounded measurable functions with domain $\statespace$. For $f\in M(\mathcal{X})$ and a probability measure $\sigma\in\prob(\mathcal{X})$, we denote $\|f\|_{2,\samplingdist}=\sqrt{\ee_{\bm x\sim \sigma}\left[f(\bm x)^2\right]}$. For a measure $\mu$ absolutely continuous with respect to another measure $\sigma$, i.e., $\mu\ll\sigma$, $\nicefrac{d\mu}{d\sigma}$ denotes the Radon-Nikodym derivative. We denote a sequence of vectors $\bm v_1,\dots, \bm v_j$ in short as $\bm v_{1:j}$. We denote the Lebesgue measure by $\lambda$, and we denote the Lebesgue measure of an $\ell^2$-ball in $d$ dimensions with radius $\eta$ by $\alpha_d\eta^d$.

We consider two-player, zero-sum, infinite-horizon discounted Markov games given by the tuple $\left(\statespace, \actionspacemin, \actionspacemax, \transitionkernel,\rwdfunc,\df\right)$, where $\statespace\subset\real^\dimstate$ is the continuous state space, and $\actionspacemin\subset\real^\dimminaction$ and $\actionspacemax\subset\real^\dimmaxaction$ are the continuous action spaces of the minimizing and maximizing players, respectively. We take $\statespace,\actionspacemin,\actionspacemax$ to be convex and compact. $\transitionkernel:(\statespace\times\actionspacemin\times\actionspacemax)\to\prob\left(\statespace\right)$ denotes the transition dynamics of the game, $\rwdfunc:(\statespace\times\actionspacemin\times\actionspacemax)\to\real$ is the immediate reward function, and $\df\in(0,1)$ is the discounting factor. Let player 1 and player 2 be the minimizing and maximizing players respectively. At each time $t\in\mathbb{Z}_{\geq0}$, at some game state $\state_t\in\statespace$, players 1 and 2 choose actions $\minaction_t\in\actionspacemin$ and $\maxaction_t\in\actionspacemax$. This yields an immediate reward of $\rwdfunc(\state_t,\minaction_t,\maxaction_t)$ and the game proceeds to the next state $\state_{t+1}\sim\transitionkernel(\cdot\vert\state_t,\minaction_t,\maxaction_t)$. Player 1 chooses actions to minimize the expected discounted cumulative return $\return=\ee\left[\sum_{t=0}^\infty\df^t\rwdfunc(\state_t,\minaction_t,\maxaction_t)\right]$, while player 2 chooses actions to maximize $\return$. 

Each player chooses its actions according to a policy. We consider deterministic Markovian policies induced by a state-action value function ($\qfunc$-function), in which players are restricted to pure strategies. For a given continuous function $\qfunc:(\statespace\times\actionspacemin\times\actionspacemax)\to\real$, at state $\state$, player 1 chooses $\pi_\qfunc(\state)\in\argmin_{\minaction\in\actionspacemin}\max_{\maxaction\in\actionspacemax}\qfunc(\state,\minaction,\maxaction)$ and player 2 chooses $\mu_\qfunc(\state)\in\argmax_{\maxaction\in\actionspacemax}\min_{\minaction\in\actionspacemin}\qfunc(\state,\minaction,\maxaction)$. These are the players' security strategies: player 1 is guaranteed a value no worse than $\min_{\minaction\in\actionspacemin}\max_{\maxaction\in\actionspacemax}\qfunc(\state,\minaction,\maxaction)$ and player 2 is guaranteed a value no worse than $\max_{\maxaction\in\actionspacemax}\min_{\minaction\in\actionspacemin}\qfunc(\state,\minaction,\maxaction)$. For such games the \emph{minimax Bellman operator} $\bellman:M\left(\statespace\times\actionspacemin\times\actionspacemax\right)\to M\left(\statespace\times\actionspacemin\times\actionspacemax\right)$ is defined by
\begin{align}\label{eq: bellman operator}  \bellman\qfunc\left(\state,\minaction,\maxaction\right)=\rwdfunc\left(\state,\minaction,\maxaction\right) + \df \ee_{\state'\sim\transitionkernel(\cdot\vert\state,\minaction,\maxaction)}\left[\inf_{\tilde{\minaction}\in\actionspacemin}\sup_{\tilde{\maxaction}\in\actionspacemax}\qfunc\left(\state',\tilde{\minaction},\tilde{\maxaction}\right)\right].
\end{align}

It is well known that the minimax Bellman operator is contractive in the $L_\infty$-norm and thus a unique fixed point $\qfunc^*$ exists such that $\bellman\qfunc^*=\qfunc^*$ \citep{shapley1953stochastic,perolat2015approximate}. \Cref{eq: bellman operator} reveals that in order to compute the Bellman operator for continuous functions, one needs to solve a \emph{minimax} problem, in this case corresponding to player 1's security strategy $\pi_\qfunc$. One can analogously define the \emph{maximin} Bellman operator associated with player 2's security strategy $\mu_\qfunc$. While both minimax and maximin problems admit solutions for compact $\actionspacemin$ and $\actionspacemax$ and continuous $\qfunc$, these solutions need not coincide, i.e., $(\pi_\qfunc,\mu_\qfunc)$ need not be a saddle-point pair for generic $\qfunc$. We restrict $\qfunc$ to a convex-concave function approximation class, convex in the player 1's action and concave in player 2's action so that the minimax theorem guarantees that at every state, $(\pi_\qfunc,\mu_\qfunc)$ forms a saddle point for every $\qfunc$ in our function class \citep{sion1958on}. Further, these saddle points can be found up to any desired accuracy using existing minimax optimization methods \citep{nemirovski2004,lin2020near}.

\begin{remark}[Mixed Strategies]
    An alternative route to ensure existence of a saddle-point solution is to let players employ mixed strategies, where they randomize over pure strategies. In the finite-action setting, mixing ensures that a saddle point of the minimax problem always exists, and it can be computed by solving a linear program at every state $\state$. In the continuous-action setting, a saddle point over mixed strategies still exists under mild continuity assumptions \citep{glicksberg1952further}, but computing it is far less straightforward than in the finite-action setting. The iterative double oracle algorithm \citep{adam2021double} is guaranteed to return a saddle point up to any desired accuracy; however, each iteration requires computing a global best response for each player, which can be intractable under function approximation. Moreover, an explicit bound on the required number of iterations as a function of the saddle-point accuracy remains unknown. These difficulties motivate our focus on pure strategies with convex-concave function approximation instead.
\end{remark}
To represent the return $\return$ obtained by policies $(\pi,\mu)$, it is useful to define the \emph{state} value function $\valfunc^{\pi,\mu}:\statespace\to\real$, given by $\valfunc^{\pi,\mu}(\state) = \ee\big[\sum_{t=0}^\infty \df^t \rwdfunc\big(\state_t,\pi(\state_t),\mu(\state_t)\big)|\state_0=\state\big]$.
The state value function $\valfunc^*$ relates to the fixed point $\qfunc^*$ of~\cref{eq: bellman operator} as $\qfunc^*(\state,\minaction,\maxaction)=\rwdfunc(\state,\minaction,\maxaction)+\df\ee_{\state'\sim\transitionkernel(\cdot|\state,\minaction,\maxaction)}\left[\valfunc^*(\state')\right]$. The state value function can be used to determine whether the policies $\pi,\mu$ form a saddle-point equilibrium via the \emph{Nash gap}. Let $\Pi$ denote the set of pure Markovian strategies, then for an evaluation distribution $\rho\in\prob(\statespace)$, the Nash gap is defined by
\begin{align}
    \texttt{NashGap}_{\rho}(\pi,\mu) = \ee_{\state\sim\rho}\left[\max_{\mu'\in\Pi}\valfunc^{\pi,\mu'}(\state) - \min_{\pi'\in\Pi}\valfunc^{\pi',\mu}(\state) \right].
\end{align}
The Nash gap is always non-negative, and is zero only when the input policies form a saddle-point Nash equilibrium with respect to $\rho$. Formally, $\texttt{NashGap}_\rho(\pi,\mu)=0\iff\valfunc^{\pi,\mu'}(\state)\le\valfunc^{\pi,\mu}(\state)\le\valfunc^{\pi',\mu}(\state)$ for all $\pi'\in\Pi$, $\mu'\in\Pi$, and $\rho$-almost every $\state\in\statespace$. In particular, when the fixed point $\qfunc^*$ of the minimax Bellman operator \labelcref{eq: bellman operator} admits pure-strategy saddle points, the induced policies $(\pi_{\qfunc^*},\mu_{\qfunc^*})$ satisfy $\texttt{NashGap}_\rho(\pi_{\qfunc^*},\mu_{\qfunc^*})=0$ for every evaluation distribution $\rho\in\prob(\statespace)$. 
The goal of this paper is to obtain bounds on the $\qfunc$-function learned from a finite number of samples, and in turn bound the Nash gap corresponding to its induced security strategies.

\section{Convex-Concave Function Approximation}\label{section: convex concave fa}
We now describe the convex-concave neural network architecture used to approximate the $\qfunc$-function: for every network parameter vector $\params\in\Theta$, and every state $\state\in\statespace$, the $\qfunc$-function network $\qfunc_\params:(\statespace\times\actionspacemin\times\actionspacemax)\to\real$ should be convex-concave, i.e., $\qfunc_\params(\state,\minaction,\maxaction)$ must be convex in $\minaction\in\actionspacemin$ for every fixed $(\state,\maxaction)$ and concave in $\maxaction\in\actionspacemax$ for every fixed $(\state,\minaction)$. 
To construct such a network, we build on \emph{partially input convex neural networks} (PICNNs) \citep{amos2017input}, a network architecture that guarantees convexity in a designated subset of the network's inputs, while simultaneously leaving the dependence on the other inputs unconstrained. Formally, given two inputs $\bm x$ and $\bm y$ where convexity in $\bm y$ is desired, the PICNN architecture constructs a neural network $N_\params(\bm x,\bm y)$, such that for every choice of parameters $\params$, $N_\params(\bm x,\bm y)$ is convex in $\bm y$ for fixed $\bm x$, while retaining full expressivity in $\bm x$. PICNNs achieve this by constraining a subset of weights along the computation graph edges traceable back to $\bm y$ to be non-negative and pairing them with convex, non-decreasing activations, while leaving weights traceable to $\bm x$ unconstrained. We review the PICNN architecture in more detail in \Cref{appendix: convex concave fa piccnn}.

To develop a \emph{convex-concave} $\qfunc_\params(\state,\minaction,\maxaction)$, we use three neural networks with parameters $\params^1,\params^2, \params^3$, representing the concatenated parameters as $\params=(\params^1;\params^2;\params^3)\in\Theta$. We use two PICNNs $N^1_{\params^1}(\state,\minaction):(\statespace\times\actionspacemin)\to\real$ (convex in the input $\minaction$) and $N^2_{\params^2}(\state,\maxaction):(\statespace\times\actionspacemax)\to\real$ (convex in the input $\maxaction$), as well as a standard feedforward matrix-valued neural network $M_{\params^3}(\state):\statespace\to\real^{\dimminaction\times\dimmaxaction}$. We construct a convex-concave $\qfunc_\params(\state,\minaction,\maxaction)$ as 
\begin{align}\label{eq: convex concave q}
\qfunc_\params(\state,\minaction,\maxaction) := N^1_{\params^1}(\state,\minaction) - N^2_{\params^2}(\state,\maxaction) + \minaction^\top M_{\params^3}(\state)\, \maxaction.
\end{align}
Note that while the first two terms $N^1_{\params^1}(\state,\minaction) - N^2_{\params^2}(\state,\maxaction)$ alone would be convex in $\minaction$ and concave in $\maxaction$, they would make $\qfunc_\params$ only additively separable in $\minaction$ and $\maxaction$. The third, bilinear term $\minaction^\top M_{\params^3}(\state)\, \maxaction$ captures interactions between the players' actions beyond the separable structure, while being simultaneously convex in $\minaction$ and concave in $\maxaction$. Note that it may be possible to include coupling terms that are more expressive than a bilinear mapping. We do not pursue this direction here, as the bilinear term suffices to guarantee joint convexity-concavity for our theoretical result in \Cref{section: online minimax fqi}, and we found it to be adequate for our experiments in \Cref{section: experiments} as well. Our theoretical results will need the the Lipschitz constant of $\qfunc_{\params}$ in the joint actions, as well as the covering number. We bound these for a particular architectural choice of $N^i$ and $M$ in \Cref{theorem: lipschitz action bound for F} of \Cref{subsection: appendix lipschitz actions bounds F}, and \Cref{theorem: covering number bound} of \Cref{appendix: convering number of convex-concave q function} respectively.

\section{Online Minimax Fitted $\qfunc$ Iteration}\label{section: online minimax fqi}
We now describe our modified minimax fitted $\qfunc$-iteration procedure, given in \Cref{algo: fqi}. We restrict the class of function approximators class to be the convex-concave neural network architecture described in \Cref{section: convex concave fa}, which we denote as $\mathcal{F}$. Further, we assume access to a minimax oracle $\oracle_\epsilon:\mathcal{F}\to(\actionspacemin\times\actionspacemax)$. For a convex-concave objective $f(\minaction,\maxaction)$ and any desired $\epsilon>0$, $\oracle_\epsilon$ returns an approximate saddle point $(\tilde{\minaction},\tilde{\maxaction})\in(\actionspacemin\times\actionspacemax)$ such that the gap $\left\vert \max_{\maxaction\in\actionspacemax}f\left(\tilde{\minaction},\maxaction\right) - \min_{\minaction\in\actionspacemin}f(\minaction,\tilde{\maxaction})\right\vert<\epsilon$. Several existing minimax optimization methods can be feasibly used as the oracle for $\qfunc\in\mathcal{F}$, because $\mathcal{F}$ guarantees convexity-concavity. At the $k^{\text{th}}$ iteration of \Cref{algo: fqi}, when the parameter estimate for the $\qfunc$-function is $\params_k$, we denote the corresponding network $\qfunc_{\params_k}$ as $\qfunck$. 
\begin{algorithm}
\caption{Online Minimax Fitted $\qfunc$-Iteration}
\begin{algorithmic}[1]
  \State \textbf{Hyper-parameters:} Number of iterations $\numiters$, batch size $n$.
  \State \textbf{Inputs:} Markov game $\left(\statespace, \actionspacemin, \actionspacemax, \transitionkernel,\rwdfunc,\df\right)$, initial $\qfunc$-function $\qfunc_0\in\mathcal{F}$, fixed sampling distribution $\samplingdist\in\prob(\statespace\times\actionspacemin\times\actionspacemax)$, minimax oracle $\oracle_\epsilon$.

  \For{$k = 0$ to $\numiters - 1$}

    \State Sample $n$ i.i.d. samples $\left\{\left(\state_i,\minaction_i,\maxaction_i\right)\right\}_{i=1}^n \sim \samplingdist$
    \State Draw $\state_i'\sim \transitionkernel\left(\cdot\mid \state_i,\minaction_i,\maxaction_i\right)$ for each $i\in[n]$
    \State For each $i\in[n]$, form the target
    \begin{gather*}
\left(\tilde{\minaction}_i,\tilde{\maxaction}_i\right)=\oracle_\epsilon\big(\qfunck\left(\state'_i,\cdot,\cdot\right)\big)\\
        y_i = \rwdfunc\left(\state_i,\minaction_i,\maxaction_i\right) + \df\qfunc_k\left(\state'_i,\tilde{\minaction}_i,\tilde{\maxaction}_i\right)
    \end{gather*}
    \State Solve the regression problem exactly
\begin{align*}
\widehat L_n(f;\qfunck)&:=\frac{1}{n}\sum_{i=1}^n\big(f(\state_i,\minaction_i,\maxaction_i)-y_i\big)^2,\\
\qquad\qfunckone &:=\arg\min_{f\in\mathcal F}\ \widehat L_n(f;\qfunck) 
\end{align*}
  \EndFor
  \State \textbf{Output:} Learned $\qfunc$-function $\qfunc_\numiters$
\end{algorithmic}
\label{algo: fqi}
\end{algorithm}

\Cref{algo: fqi} differs from standard fitted $\qfunc$-iteration in the sense that we use data in an \emph{online} manner and do not reuse it. Specifically, at iteration $k$, using a fixed sampling distribution $\samplingdist\in\prob(\statespace\times\actionspacemin\times\actionspacemax)$, we sample $n$ i.i.d. state-action tuples $(\state_i,\minaction_i,\maxaction_i)\sim\samplingdist$ for $i\in[n]$. Using these  states and actions, we sample the next states according to the transition dynamics $\state'_i\sim\transitionkernel(~\cdot~\vert\state_i,\minaction_i,\maxaction_i)$ for $i\in[n]$. For each $i\in[n]$, we invoke the minimax oracle $\oracle_\epsilon$ on the fitted $\qfunc$-function $\qfunck\in\mathcal{F}$ evaluated at the next state $\state_i'$ to obtain an approximate saddle-point action pair $(\tilde{\minaction}_i,\tilde{\maxaction}_i) := \oracle_\epsilon\big(\qfunck(\state_i',\cdot,\cdot)\big)$. Since $\qfunck$ is convex-concave from \Cref{eq: convex concave q}, the minimax theorem guarantees that a saddle point of $\qfunck(\state_i',\cdot,\cdot)$ exists, and $\oracle_\epsilon$ returns an $\epsilon$-accurate approximation to it. We use $(\tilde{\minaction}_i,\tilde{\maxaction}_i)$ to form the regression target $y_i = \rwdfunc(\state_i,\minaction_i,\maxaction_i) + \df\qfunck(\state_i',\tilde{\minaction}_i,\tilde{\maxaction}_i)$, which serves as an empirical estimate of an \emph{approximate} minimax Bellman operator from \Cref{eq: bellman operator} applied to the sampled tuple $(\state_i,\minaction_i,\maxaction_i)$. This approximation occurs because in practice, one can solve the $\min\max$ problem up to $\epsilon$-accuracy instead of exactly. Finally, we update the $\qfunc$-function estimate by solving a least-squares regression problem over $\mathcal{F}$:
\begin{equation}\label{eq: fqi regression}
    \qfunckone := \arg\min_{f\in\mathcal{F}} \widehat{L}_n(f;\qfunck) := \arg\min_{f\in\mathcal F}\frac{1}{n}\sum_{i=1}^n\left(f(\state_i,\minaction_i,\maxaction_i)-y_i\right)^2.
\end{equation}
Note that in \Cref{eq: fqi regression}, the regression target~$y_i$ is treated as a constant independent of $\mathcal{F}$. Repeating this procedure for $\numiters$ iterations yields the final estimate $\qfunc_{\numiters}$ returned by \Cref{algo: fqi}.

\paragraph{Sources of error in \Cref{algo: fqi}.} We now shift our focus to analyzing how the learned $\qfunc_\numiters$ compares to the fixed point of the minimax Bellman operator $\qfunc^*$. It is pertinent to note that \Cref{algo: fqi} departs from an exact Bellman iteration $\qfunc_{k+1} = \bellman\qfunc_k$ in three ways, which introduce three corresponding sources of error: (i) an optimization error from only being able to solve the minimax problem up to $\epsilon$-accuracy with $\oracle_\epsilon$, (ii) an approximation error from restricting to the convex-concave function class $\mathcal{F}$, and (iii) a statistical error from using finitely many samples per iteration. We now describe each in detail. First, because we solve the minimax problem in \Cref{eq: bellman operator} with $\epsilon$-accuracy and not exactly, \Cref{algo: fqi} is better represented by the \emph{approximate} minimax Bellman operator $\approxbellman:M\left(\statespace\times\actionspacemin\times\actionspacemax\right)\to M\left(\statespace\times\actionspacemin\times\actionspacemax\right)$ defined by
\begin{align}\label{eq: approx bellman operator}
    \approxbellman\qfunc(\state,\minaction,
    \maxaction)= \rwdfunc(\state,\minaction,
    \maxaction) + \df\ee_{\state'\sim\transitionkernel\left(\cdot\vert\state,\minaction,\maxaction\right)}\left[\qfunc\left(\state', \oracle_\epsilon\big(\qfunc\left(\state',\cdot,\cdot\right)\right)\big)\right],
\end{align}
where use a slight abuse of notation by denoting $\qfunc(\state',\tilde{\minaction}, \tilde{\maxaction})$ where $(\tilde{\minaction}, \tilde{\maxaction})=\oracle_\epsilon\big(\qfunc(\state',\cdot,\cdot)\big)$ as $\qfunc\big(\state', \oracle_\epsilon(\qfunc(\state',\cdot,\cdot))\big)$.
Secondly, without additional assumptions on the Markov game, the approximate Bellman operator need not satisfy \emph{Bellman completeness}, i.e., it is possible that for $\qfunc\in\mathcal{F}$, we have $\approxbellman\qfunc\notin\mathcal{F}$. This motivates the fundamental approximation error
\begin{align}\label{eq: delapprox}
    \delapprox := \sup_{\qfunc\in\mathcal{F}}\inf_{f\in\mathcal{F}}\left\|f-\approxbellman\qfunc\right\|_{2,\samplingdist}.
\end{align}
The third error is statistical: because we only have $n$ samples at each iteration, we can only perform an empirical approximate Bellman iteration. We will show how these affect the bound on $\|\qfunc_\numiters-\qfunc^*\|_{2,\rho}$ for a distribution $\rho\in\prob(\statespace\times\actionspacemin\times\actionspacemax)$, where $\qfunc^*$ is the fixed point of the \emph{exact} Bellman operator $\bellman$ given in \Cref{eq: bellman operator}. We now list our assumptions that hold throughout our analysis.

\begin{assumption}[State-Action Spaces]\label{assumption: actions}
    The continuous state and action spaces $\statespace, \actionspacemin$ and $\actionspacemax$ are compact and convex, with $D_\statespace := \sup_{\state\in\statespace}\|\state\|_2<\infty$, $D_\actionspacemin := \sup_{\minaction\in\actionspacemin}\|\minaction\|_2<\infty$ and $D_\actionspacemax:= \sup_{\maxaction\in\actionspacemax}\|\maxaction\|_2
    <\infty$.
\end{assumption}
\begin{assumption}[Reward Function]\label{assumption: reward}
    The immediate reward function $\rwdfunc$ is continuous and $\|\rwdfunc\|_\infty\leq \rmax<\infty$. Further, $\rwdfunc$ is $L_\rwdfunc$-Lipschitz in its actions, i.e., there exists $L_\rwdfunc>0$ such that for all 
    $\state\in\statespace$, $\minaction,\minaction'\in\actionspacemin$, and $\maxaction,\maxaction'\in\actionspacemax$,
    \begin{align*}
        |\rwdfunc(\state,\minaction,\maxaction)-\rwdfunc(\state,\minaction',\maxaction')|\leq L_\rwdfunc\left(\|\minaction-\minaction'\|_2+\|\maxaction-\maxaction'\|_2\right).
    \end{align*}
\end{assumption}
\begin{assumption}[Transition Dynamics]\label{assumption: dyn}
    For every bounded continuous function $h:\mathcal S\to\mathbb R$, the map $(\state,\minaction,\maxaction)\mapsto \mathbb E_{s'\sim P(\cdot\mid \state,\minaction,\maxaction)}[h(\state')]$ is continuous. Further, $\transitionkernel$ is also $L_\transitionkernel$-Lipschitz in its actions, i.e.,, there exists $L_\transitionkernel>0$ such that for all 
    $\state\in\statespace$, $\minaction,\minaction'\in\actionspacemin$, $\maxaction,\maxaction'\in\actionspacemax$, and measurable sets $S$ of $\statespace$,
    \begin{align*}
        P(S|\state,\minaction,\maxaction)-P(S|\state,\minaction',\maxaction') \leq L_\transitionkernel\left(\|\minaction-\minaction'\|_2+\|\maxaction-\maxaction'\|_2\right).
    \end{align*}
\end{assumption}
\Cref{assumption: actions,assumption: reward,assumption: dyn} are standard in the literature and ensure that the Bellman operator $\bellman$ preserves continuity, admits a fixed point, and is Lipschitz continuous.
\begin{assumption}[Exploratory Distribution]\label{assumption: exploratory dist}
    There exists $c>0$ such that exploratory sampling distribution $\samplingdist$ has a conditional density $\condidens(\minaction,\maxaction|\state)$ satisfying $\condidens(\minaction,\maxaction|\state)\geq c>0$ for all $\state\in\statespace$, $\minaction\in\actionspacemin$, $\maxaction\in\actionspacemax$.
\end{assumption}

\begin{assumption}[Action Spaces Regularity]\label{assumption: action space regularity}
    Consider a ball $\mathcal{B}(\bm a, \eta)$ in the joint action space $\actionspacemin\times\actionspacemax\subset\real^{\dimminaction+\dimmaxaction}$ centered at $\bm a\in\actionspacemin\times\actionspacemax$ with radius $\eta>0$. We assume there exists $\alpha>0$ such that
    \begin{align*}
        \lambda\big(\mathcal{B}(\bm a, \eta)\cap (\actionspacemin\times\actionspacemax)\big) \geq \min\left(\alpha\cdot\alpha_{\dimminaction+\dimmaxaction}\eta^{\dimminaction+\dimmaxaction}, \lambda\left(\actionspacemin\times\actionspacemax\right)\right).
    \end{align*}
\end{assumption}
\Cref{assumption: exploratory dist} ensures that the action sampling distribution assigns positive mass to all regions of the players' action spaces, and is satisfied by standard distributions used in reinforcement learning such as Gaussian and uniform distributions. \Cref{assumption: action space regularity} is a mild geometric condition on the players' joint action spaces ruling out pathological geometries near the boundary, and is satisfied by convex action spaces.
\begin{assumption}[Convex-Concave Function Approximation]\label{assumption: FA}
    For the class of function approximators $\mathcal{F}$, every $\qfunc\in\mathcal F$ satisfies:
    \begin{itemize}
        \item Boundedness: $\|\qfunc\|_\infty \leq \qmax$, with $\qmax\geq \nicefrac{\rmax}{1-\df}$.
        \item Convex-concavity: $\qfunc(\state,\cdot,\cdot)$ is convex in $\minaction$ for fixed $\maxaction$ and concave in $\maxaction$ for fixed $\minaction$, for every $\state$.
        \item Joint Lipschitzness: There exists a positive constant $\lipschitz<\infty$ such that for all $\qfunc\in\mathcal{F}$,  $|\qfunc(\state,\minaction,\maxaction) - \qfunc(\state,\minaction',\maxaction')| \leq \lipschitz\big(\|\minaction-\minaction'\|_2 + \|\maxaction-\maxaction'\|_2\big)~\forall~\state$ and all $\minaction,\minaction'\in\actionspacemin$, $\maxaction,\maxaction'\in\actionspacemax$.
    \end{itemize}
\end{assumption}
Convexity-concavity in \Cref{assumption: FA} is satisfied by the neural network architecture we introduce in \Cref{section: convex concave fa} through construction. The Lipschitzness can be ensured by standard machine learning practices such as spectral normalization with bounded inputs \citep{miyato2018spectral}. Note that the same practices also ensure boundedness. We provide a Lipschitz bound for $\mathcal{F}$ in \Cref{theorem: lipschitz action bound for F} of \Cref{subsection: appendix lipschitz actions bounds F}.
\begin{assumption}[Minimax Oracle]\label{assumption: minimax oracle}
    We have access to an oracle $\oracle_\epsilon:\statespace\times\mathcal{F}\to\actionspacemin\times\actionspacemax$, such that for any desired $\epsilon>0$, for $(\pi_\qfunc,\mu_\qfunc):=\oracle_\epsilon(\state,\qfunc)$ we have
    \begin{equation*}
        \left\vert \max_{\maxaction\in\actionspacemax}\qfunc\left(\state,\pi_\qfunc,\maxaction\right) - \min_{\minaction\in\actionspacemin}\qfunc\left(\state,\minaction,\mu_\qfunc\right)\right\vert < \epsilon~\forall~\state\in\statespace,\qfunc\in\mathcal{F}.
    \end{equation*}
\end{assumption}
\Cref{assumption: minimax oracle} is satisfied for convex-concave, Lipschitz $\qfunc$-functions by a variety of minimax optimization methods introduced in prior literature \citep{nemirovski2004,lin2020near}.
\begin{assumption}[Discounted Concentrability Coefficients]\label{assumption: discounted conc coeff}
For the sampling distribution $\samplingdist\in\prob(\statespace\times\actionspacemin\times\actionspacemax)$ with state marginal 
    $\samplingdist_{\statespace}$, let $\pi^j$ 
    denote any Markov policy with $\pi^j : \statespace \to 
    \actionspacemin \times \actionspacemax$. Define the $j^{\mathrm{th}}$-step 
    concentrability coefficient for a distribution $\nu\ll\samplingdist_\statespace$
    \begin{align*}
        c_j(\nu) := \sup_{\pi^1,\dots,\pi^j}
        \left\|\frac{d(\nu P^{\pi^1}\cdots P^{\pi^j})}
        {d\samplingdist_{\statespace}}\right\|_{\infty}, \quad c_0(\nu)=\left\|\frac{d\nu}{d\samplingdist_\statespace}\right\|_\infty,\quad C_{\nu} := (1-\df)\sum_{j=0}^{\infty} \df^{j} c_j(\nu),
    \end{align*}
   with $P^{\pi^j}(\cdot|\state):=P(\cdot|\state,\pi^j(\state))$. Then, we assume that $C_{\sigma_\statespace}<\infty$.
\end{assumption}
\Cref{assumption: discounted conc coeff} quantifies the similarity between the sampling distribution $\sigma$ and the future state distributions of the Markov game, and analogous versions of \Cref{assumption: discounted conc coeff} are standard in reinforcement learning literature for both Markov decision processes and zero-sum Markov games \citep{antos2007fitted,munos2008finite,scherrer2015approximate,lazaric2016analysis,fan2020theoretical}.

To analyze the convergence of \Cref{algo: fqi}, we assume that $\qfunc^*$ from \Cref{section: prelims} admits a pure-strategy saddle point, so that \Cref{algo: fqi} converges towards a meaningful equilibrium. We formally state this as follows.
\begin{assumption}[Saddle-point Existence]\label{assumption: saddle point exists}
    For every $s\in\statespace$, $\min_{\minaction\in\actionspacemin}\max_{\maxaction\in\actionspacemax}\qfunc^*(\state,\minaction,\maxaction)=\max_{\maxaction\in\actionspacemax}\min_{\minaction\in\actionspacemin}\qfunc^*(\state,\minaction,\maxaction)$.
\end{assumption} 

We are now ready to state our main result. For the sampling distribution $\samplingdist\in\prob(\statespace\times\actionspacemin\times\actionspacemax)$, $\samplingdist_\statespace$ denotes the marginal of $\samplingdist$ over $\statespace$.
\begin{theorem}\label{theorem: fqi bounds}
    Under \Cref{assumption: actions,assumption: reward,assumption: dyn,assumption: exploratory dist,assumption: action space regularity,assumption: FA,assumption: minimax oracle,assumption: discounted conc coeff,assumption: saddle point exists}, let an evaluation distribution $\rho\in\prob(\statespace)$ satisfy $c_0(\rho)=\|\nicefrac{d\rho}{d\samplingdist_\statespace}\|_\infty\leq C_{\rho,\samplingdist}<\infty$. Then for all $\delta>0$, the output $\qfunc_\numiters$ of \Cref{algo: fqi} after $\numiters$ iterations satisfies, with probability at least $1-\delta$
    \begin{align*}
   &\ee_{\state\sim\rho}\left[\sup_{\minaction\in\actionspacemin,\maxaction\in\actionspacemax}\left|\qfunc_\numiters-\qfunc^*\right|\right] \leq 2\qmax C_{\rho,\samplingdist}\cdot\df^\numiters + \frac{\sqrt{6}C_{\samplingdist_{\statespace}} C_{\rho,\samplingdist}\left(L_\mathcal{F}+L_\rwdfunc + 2\df\qmax L_\transitionkernel\right)\eta_2}{1-\df}\\
   &+\frac{2C_{\samplingdist_{\statespace}} C_{\rho,\samplingdist}}{(1-\df)\sqrt{c\volume\left(\nicefrac{\eta_2}{2}\right)}} \left(\gamma\epsilon + \sqrt{\delapprox^2+ 5\Delta_n\left(\eta_1,\frac{\delta}{3\numiters}\right) + 4\qmax^2\sqrt{\frac{\log\left(\nicefrac{3\numiters}{\delta}\right)}{2n}} + 6\qmax\eta_1}\right),
    \end{align*}
where 
\begin{align*}
\Delta_n\left(\eta_1,\nicefrac{\delta}{3\numiters}\right):=4\qmax^2\sqrt{\frac{2\log\left(\nicefrac{6K\coveringnumber(\eta_1)}{\delta}\right)}{n}} + \frac{8\qmax^2\log\left(\nicefrac{6\numiters\coveringnumber(\eta_1)}{\delta}\right)}{3n} + 2\qmax\eta_1,
\end{align*}
$V(\nicefrac{\eta_2}{2})=\inf_{\bm a\in\actionspacemin\times\actionspacemax}\lambda\Big(\mathcal{B}(\bm a, \nicefrac{\eta_2}{2})\cap (\actionspacemin\times\actionspacemax)\Big)\geq \min\left(\alpha\alpha_{\dimminaction+\dimmaxaction}\left(\nicefrac{\eta_2}{2}\right)^{\dimminaction+\dimmaxaction}, \lambda\left(\actionspacemin\times\actionspacemax\right)\right)$ from \Cref{assumption: action space regularity}, $N(\eta_1):=N(\mathcal{F},\eta_1,\|\cdot\|_\infty)$ is the covering number for the class of functions $\mathcal{F}$, and $\eta_1,\eta_2>0$ are chosen positive constants, and $\delapprox$ is defined in \Cref{eq: delapprox}. 
\end{theorem}

The proof of \Cref{theorem: fqi bounds} is in Appendix~\ref{section: appendix fqi proof}. \Cref{theorem: fqi bounds} decomposes the error $\sup_{\minaction\in\actionspacemin,\maxaction\in\actionspacemax}\|\qfunc_\numiters-\qfunc^*\|$ for some evaluation distribution $\rho$ into several error sources. The first term is an algorithmic error that decays geometrically. The remaining terms capture the three error sources discussed above, and do not vanish with $\numiters$, representing the fundamental difficulty of the problem. The $\epsilon$ term reflects the cost of using an $\epsilon$-accurate solver as exact minimax solutions cannot be practically computed, and this term linearly vanishes with $\epsilon$ at the cost of more costly oracle computation. The presence of $\delapprox^2$ in the bound represents the cost of working with a restricted, tractable convex-concave function class $\mathcal{F}$. $\delapprox$ will be $0$ when the underlying game satisfies Bellman completeness. 

The terms $\mathcal{O}\big(\Delta_n(\eta_1,\nicefrac{\delta}{3K})\big)$ and $\mathcal{O}\big(\sqrt{\nicefrac{\log\left(\nicefrac{3K}{\delta}\right)}{2n}}\big)$ correspond to the statistical error and quantify how fast the statistical error decays with the number of samples $n$ and the covering number $N(\eta_1)$ of $\mathcal{F}$. Unlike the finite-action setting, continuous action spaces cannot be covered exactly with finitely many samples. The parameters $\eta_1,\eta_2>0$ correspond to this covering tradeoff. Smaller $\eta_1,\eta_2$ lead to better approximation but higher covering number $N(\eta_1)$ and larger inverse volume $\nicefrac{1}{\sqrt{V\left(\nicefrac{\eta_2}{2}\right)}}$, which as mentioned in \Cref{theorem: fqi bounds}, depends exponentially on the dimension of the action spaces. The covering number $N(\eta_1)$ depends on the particular architecture used to construct $\qfunc\in\mathcal{F}$ and depends polynomially on the action space dimensions. We give a bound on $N(\eta_1)$ for a particular architecture in \Cref{theorem: covering number bound} in \Cref{appendix: convering number of convex-concave q function}.

To the best of our knowledge, \Cref{theorem: fqi bounds} is the first finite-sample guarantee of this form for zero-sum Markov games with continuous state and action spaces beyond the linear quadratic setting. \Cref{theorem: fqi bounds} also allows us to establish a bound on the Nash gap when players employ their security strategies $\left(\pi_{\qfunc_\numiters},\mu_{\qfunc_\numiters}\right)$ induced by the learned $\qfunc$-function $\qfunc_\numiters$.
\begin{corollary}\label{corollary: nash gap}
    Let $\rho\in\prob(\statespace)$ be an evaluation distribution. Under \Cref{assumption: actions,assumption: reward,assumption: dyn,assumption: exploratory dist,assumption: action space regularity,assumption: FA,assumption: minimax oracle,assumption: discounted conc coeff,assumption: saddle point exists}, if $C_\rho<\infty$, the oracle-returned approximate security strategies $\left(\pi_{\qfunc_{\numiters}}(\state),\mu_{\qfunc_{\numiters}}(\state)\right)=\oracle_\epsilon\big(\state,\qfunc_\numiters\big)$ corresponding to the output $\qfunc_\numiters$ of \Cref{algo: fqi} satisfy
    \begin{align*}
\textnormal{\texttt{NashGap}}_\rho(\pi_{\qfunc_\numiters},\mu_{\qfunc_\numiters})\leq\frac{4C_\rho}{1-\df}\ee_{\state\sim\sigma_\statespace}\left[\sup_{\minaction\in\actionspacemin,\maxaction\in\actionspacemax}|\qfunc_K-\qfunc^*|\right] +\frac{2\epsilon}{1-\df}.
    \end{align*}
\end{corollary}
The proof of \Cref{corollary: nash gap} is in \Cref{appendix: nash gap proof}. Importantly, \Cref{corollary: nash gap} shows how the bounds developed in \Cref{theorem: fqi bounds} can be used to bound the Nash gap itself.

\section{Simulation Experiments}\label{section: experiments}
We now complement our result from \Cref{theorem: fqi bounds} with simulation experiments investigating whether the convex-concave $\qfunc$-function neural network architecture introduced in \Cref{eq: convex concave q} leads to stable training in practice. The practical issue arises because one cannot implement \Cref{algo: fqi} exactly, since each iteration requires solving a non-convex regression problem of optimizing a neural network. We therefore empirically study a practical version of \Cref{algo: fqi} in which at each iteration, we take a gradient step on the empirical loss, mimicking online temporal difference (TD) learning. At iteration $k$ of \Cref{algo: fqi}, instead of getting $\qfunckone$ from \Cref{eq: fqi regression}, we have for a step-size $\alpha>0$
\begin{align*}
\params_{k+1} = \params_k - \alpha \nabla_\params\hat{L}_n(\qfunck;\qfunck), \quad\qfunckone = \qfunc_{\params_{k+1}}.
\end{align*}
We investigate whether this practical modification to \Cref{algo: fqi} trains stably with the convex-concave $\qfunc$-function architecture of \Cref{eq: convex concave q}.

\textbf{Markov Game Setup.} We consider a planar pursuit-evasion game between two players in a $8$-dimensional state space $\statespace=[-2.5,2.5]^8$. At a state $\state\in\statespace$, the first $4$ components $(s_1,s_2), (s_3,s_4)$ denote player 1's planar position and velocity, and the last $4$ components $(s_5,s_6), (s_7,s_8)$ denote player 2's planar position and velocity. Both players have $2$-dimensional action spaces $\actionspacemin,\actionspacemax=[-1, 1]^2$, and actions denote the players' horizontal and vertical planar accelerations. Both players follow noisy double integrator dynamics $\state'=A\state + B^1\minaction+B^2\maxaction + w_n\xi$, where $A = \textrm{diag}(A_p, A_p)$, $B^1=\begin{bmatrix}
    B_p;0
\end{bmatrix}$, $B^2=\begin{bmatrix}
    0;B_p
\end{bmatrix}$,  $\xi\sim\mathcal{N}(0,I)$, 
\begin{align*}
    A_p = \begin{bmatrix}
    1 &0 & \Delta t &0\\
    0 &1 &0 &\Delta t\\
    0 &0& 1& 0\\
    0 &0 &0& 1
    \end{bmatrix},
    \text{ and } B_p = 
    \begin{bmatrix}
         \frac{1}{2}\Delta t^2& 0\\
    0& \frac{1}{2}\Delta t^2\\
    \Delta t& 0\\
    0 &\Delta t
    \end{bmatrix}.
\end{align*}
Player 1 wishes to pursue player 2, who evades player 1, while both want the save their control effort and waste the other's control effort. We encode this behavior with the immediate reward function

\begin{align*}
\rwdfunc(\state,\minaction,\maxaction) = -w\exp\left(-\frac{\left((s_1-s_5)^2 + (s_2-s_6)^2\right)}{2w_d^2}\right) + \frac{1}{2}\minaction^\top R^1\minaction - \frac{1}{2}\maxaction^\top R^2\maxaction.
\end{align*}
Note that this game lies beyond the linear-quadratic setting studied in prior work \citep{al2007adaptive,al2007model,zhang2019policy,zhang2021provably,bu2019global}, and the optimal $\qfunc$-function cannot be found by using Riccati equations. We choose the sampling distribution $\samplingdist=\big(\text{Uni}\left([-2.5,2.5]^8\right)\times \text{Uni}\left([-1,1]^2\right)\times \text{Uni}\left([-1,1]^2\right) \big)$, with independent marginals. The discount factor $\df$ is $0.95$. We set $w=0.01,w_d=0.1,w_n=0.01,\Delta t=0.1, R^1=0.01I$ and $R^2=0.02 I$. We use SecOND \citep{gupta2024second}, which is a second-order constraint-aware solver, as our minimax oracle $\oracle_\epsilon$, as we found it to converge favorably in practice. Instead of setting an $\epsilon$, we terminate the solver once successive candidate solution iterates differ by less than a tolerance of $10^{-8}$. For convex-concave, Lipschitz $\qfunc$-functions, this serves as a practical proxy for $\epsilon$.
We use the partially input convex-concave neural network as described in \Cref{section: convex concave fa} for our choice of $\qfunc$-function, with softplus activations. We empirically bound the value that the $\qfunc$-function neural network outputs by clipping values to lie between $[\nicefrac{\rwdfunc_{\text{min}}}{1-\df},  \nicefrac{\rwdfunc_{\text{max}}}{1-\df}]$, and found that while in theory this breaks convexity-concavity, it did not harm performance in practice. Additional implementation details are provided in \Cref{appendix: experimental details}.

Along the iterations $k$, \Cref{fig: mean absolute error} plots the normalized mean absolute error 
\begin{align*}
    \delta_{\text{NMAE}}:=\frac{\nicefrac{1}{n} \sum_{i=1}^n \left|\qfunck(\state_i,\minaction_i,\maxaction_i) - y_i\right|}{\nicefrac{\left(\rwdfunc_{\text{max}}-\rwdfunc_{\text{min}}\right)}{1-\df}},
\end{align*}
where $y_i$ is the regression target formed at iteration $k$ of \Cref{algo: fqi} for the $i^{\text{th}}$ data sample pair $(\state_i,\minaction_i,\maxaction_i), i\in[n]$, and the denominator $\nicefrac{\left(\rwdfunc_{\text{max}}-\rwdfunc_{\text{min}}\right)}{1-\df}$ is the largest absolute error attainable given the bounded $\qfunc$-value range derived above.

\Cref{fig: mean absolute error} shows that the TD-style modification of \Cref{algo: fqi} trains stably with the partially input convex-concave $\qfunc$-function architecture.  The error $\delta_{\text{NMAE}}$ decreases and plateaus within the first $100$ iterations across all sample sizes we test. Across all 20 seeds, there is no divergence observed across the iterations. This empirically suggests that using the partially input convex-concave architecture from \Cref{eq: convex concave q} for the $\qfunc$-function allows stable training. Though \Cref{theorem: fqi bounds} does not directly apply to this practical version, we observe that both the error and its spread shrink as $n$ increases from $100$ to $500$, which is consistent with the statistical benefit of larger per-iteration sample size $n$.

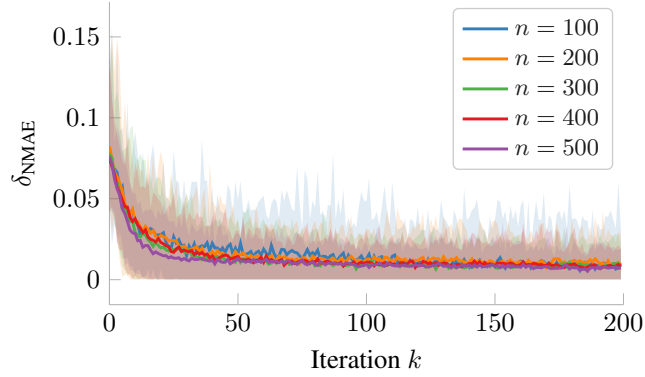
\begin{figure}[t]
\centering
\begingroup
\catcode`\"=9\relax %
\begin{tikzpicture}
\begin{axis}[
    width=0.6\linewidth,
    height=0.4\linewidth,
    xlabel={Iteration $k$},
    ylabel={$\delta_{\text{NMAE}}$},
    axis x line*=bottom,
    axis y line*=left,
    axis line style={gray!60},
    tick style={gray!60},
    ymajorgrids=false,
    yticklabel style={
        /pgf/number format/fixed,
        /pgf/number format/precision=2
    },
    scaled y ticks=false,
    legend style={
        at={(0.98,0.98)},
        anchor=north east,
        fill=white,
        fill opacity=0.9,
        draw=gray!40,
        line width=0.5pt,
        rounded corners=2pt,
        font=\small,
        /tikz/every even column/.append style={column sep=4pt},
    },
    legend cell align={left},
    xmin=0, xmax=200,
]
\addlinewithband{sections/figures/ns100_new.csv}{c1}{$n=100$}
\addlinewithband{sections/figures/ns200_new.csv}{c2}{$n=200$}
\addlinewithband{sections/figures/ns300_new.csv}{c3}{$n=300$}
\addlinewithband{sections/figures/ns400_new.csv}{c4}{$n=400$}
\addlinewithband{sections/figures/ns500_new.csv}{c5}{$n=500$}

\end{axis}
\end{tikzpicture}
\endgroup
\caption{\textbf{The TD-style modification of \Cref{algo: fqi} trains stably with the partially input convex-concave $\qfunc$-function architecture.} Each curve shows the normalized mean absolute training error $\delta_{\text{NMAE}}$ against iteration $k$, for $n$ samples per iteration. Solid lines are the mean over $20$ random seeds and shaded bands show the min--max spread across seeds. An increase in $n$ yields both lower steady-state error and lower variance.}
\label{fig: mean absolute error}
\end{figure}

\section{Conclusion}
We study the learning of $\qfunc$-functions in infinite-horizon zero-sum Markov games in continuous state and action spaces. Prior finite-sample guarantees of $\qfunc$-function learning have largely been confined to either games with finite-actions or to continuous linear-quadratic games, where the agents have linear dynamics and quadratic rewards. A key challenge in extending the guarantees of these prior works to more general continuous-action zero-sum Markov games is that the associated Bellman equation involves a minimax problem that is not guaranteed to admit a tractable saddle-point solution in pure strategies. To address this issue, we introduce a class of neural network function approximation for $\qfunc$ which is convex-concave in the players' actions, guaranteeing that the minimax problem admits a pure-strategy saddle point that can be computed efficiently by existing solvers. We study an online variant of the fitted $\qfunc$-iteration algorithm employing this function class and establish, to the best of our knowledge, the first finite-sample guarantees for continuous state-action zero-sum Markov games beyond the linear-quadratic setting. Finally, we empirically validate that these convex-concave neural networks lead to stable training when used on a practical version of modified fitted $\qfunc$-iteration. The empirical results naturally point to a future work direction of extending the finite-sample analysis for the practical temporal difference-style variant of \Cref{theorem: fqi bounds} studied in \Cref{section: experiments}, in which the exact regression step of \Cref{algo: fqi} is replaced by a single gradient step, as in \Cref{section: experiments}.

\subsection*{AI use statement}
In this work, we used generative AI tools for the following critical tasks:
\begin{itemize}
    \item Assist in a literature review about zero-sum continuous state-action Markov games.
    \item Help in understanding the proof machinery of related papers, to see how can existing mathematical techniques be used in our work. These  techniques were critical for proving our mathematical claims.
    \item Help in implementing the practical version of \Cref{algo: fqi} for experiments in \Cref{section: experiments}. 
\end{itemize}
We have not used generative AI tools for other tasks with required disclosure mentioned in the AI policy for authors like proof writing and developing conceptual frameworks. The rest of the required disclosure tasks are not applicable to this work.
Additionally, we used generative AI tools for checking grammatical correctness of our writing. We have reviewed all AI-assisted work. For the tasks listed above, we took the following measures:
\begin{itemize}
    \item Any relevant paper suggested by an AI tool was manually screened to verify its relevance to our work.
    \item The applicability of proof techniques from prior work that were suggested by AI tools were manually verified, and the proof was then manually done using these suggested techniques.
    \item LLM-generated implementation was manually verified and tested.
\end{itemize}
We take responsibility for the final content of this work,
including text, claims or artifacts produced with the aid of generative AI.
\subsection*{Ethics statement}
Our work deals with zero-sum game theory. Zero-sum games can be used for robust, safety-oriented applications, but the same techniques may be misused adversarially. Our work is mostly theoretical and does not directly enable any specific application that could be exploited in this manner.

\subsection*{Reproducibility statement}
This work is mostly theoretical and all assumptions for the main \Cref{theorem: fqi bounds} are listed before the theorem statement. The proofs of \Cref{theorem: fqi bounds} and \Cref{corollary: nash gap} are given in \Cref{section: appendix fqi proof} and \Cref{appendix: nash gap proof} respectively. Additional bounds and corresponding proofs are presented in \Cref{subsection: appendix lipschitz actions bounds F} and \Cref{appendix: convering number of convex-concave q function}. For our experiments, all details are presented in \Cref{section: experiments} and \Cref{appendix: experimental details}. Code with running instructions can be found in the supplementary material.

\subsubsection*{Acknowledgments}
DISTRIBUTION STATEMENT A. Approved for public release. Distribution is unlimited.

This material is based upon work supported by the Under Secretary of War for Research and Engineering under Air Force Contract No. FA8702-15-D-0001 or FA8702-25-D-B002. Any opinions, findings, conclusions or recommendations expressed in this material are those of the author(s) and do not necessarily reflect the views of the Under Secretary of War for Research and Engineering.

© 2026 Massachusetts Institute of Technology.

Delivered to the U.S. Government with Unlimited Rights, as defined in DFARS Part 252.227-7013 or 7014 (Feb 2014). Notwithstanding any copyright notice, U.S. Government rights in this work are defined by DFARS 252.227-7013 or DFARS 252.227-7014 as detailed above. Use of this work other than as specifically authorized by the U.S. Government may violate any copyrights that exist in this work.

Additionally, this work was also supported by the following: Army Research Office under Grant No. W911NF-23-1-0317, Office of Naval Research under Grant No. N00014-24-1-2797 and a National Science Foundation CAREER award under Grant No. 2336840

\bibliography{iclr2027_conference}
\bibliographystyle{iclr2027_conference}

\appendix
\section{Proof of \Cref{theorem: fqi bounds}}\label{section: appendix fqi proof}
We will first list several classical results. Then, to prove \Cref{theorem: fqi bounds}, we establish a per-iteration error bound on $\qfunckone - \bellman \qfunck$ in Appendix~\ref{subsection: per iteration error bound}, and in 
Appendix~\ref{subsection: error prop across iterations} show how this bound propogates across iterations and leads to the final result of \Cref{theorem: fqi bounds}.

\subsection{Useful Concentration Inequalities}
We list two common concentration inequalities here that we use. A detailed background can be found in standard references, for example, \citet{Vershynin_2026}.
\begin{fact}[Bernstein]
    If $Z_1,\dots,Z_n$ are i.i.d., $\mathbb E[Z_i]=0$, $|Z_i|\le b$ a.s., $\mathrm{Var}(Z_i)\le\tau^2$, then for every $t>0$,
$$
\mathbb P\left(\frac1n\sum_{i=1}^n Z_i > t\right) \;\le\; \exp\left(-\frac{nt^2/2}{\tau^2+bt/3}\right).
$$
Equivalently, for every $\delta\in(0,1)$, with probability $\ge1-\delta$, we have 
\begin{align*}
    \frac1n\sum_{i=1}^n Z_i \le \sqrt{\frac{2\tau^2\log(1/\delta)}{n}} + \frac{2b\log(1/\delta)}{3n}.
\end{align*}
\end{fact}
\begin{fact}[Hoeffding]
    If $Z_i,i\in[n]$ are bounded random variables such that $a_i\leq Z_i\leq b_i$ almost surely, then $\forall~t>0$
    \begin{align*}
        P\left(\frac{1}{n}\sum_{i=1}^n Z_i -\frac{1}{n}\sum_{i=1}^n\ee\left[ Z_i\right]\right)\leq \exp\left(-\frac{2n^2t^2}{\sum_{i=1}^n(b_i-a_i)^2}\right).
    \end{align*}
    Equivalently, for $\delta>0$, we have with probability $\geq1-\delta$
    \begin{align*}
        \frac{1}{n}\sum_{i=1}^n Z_i\leq \frac{1}{n}\sum_{i=1}^n\ee\left[Z_i\right] + \sqrt{\frac{\sum_i(b_i-a_i)^2\log\left(\nicefrac{1}{\delta}\right)}{2n^2}}.
    \end{align*}
\end{fact}
\subsection{Per-iteration Error Bound}\label{subsection: per iteration error bound}
We will first bound the per-iteration error incurred by \Cref{algo: fqi}. For convenience we restate the minimax Bellman operator $\bellman$ from \Cref{eq: bellman operator} and approximate minimax Bellman operator $\approxbellman$ from \Cref{eq: approx bellman operator}:
\begin{align*}
    \bellman\qfunc\left(\state,\minaction,\maxaction\right)=\rwdfunc\left(\state,\minaction,\maxaction\right) + \df \ee_{\state'\sim\transitionkernel(\cdot\vert\state,\minaction,\maxaction)}\left[\min_{\tilde{\minaction}\in\actionspacemin}\max_{\tilde{\maxaction}\in\actionspacemax}\qfunc\left(\state',\tilde{\minaction},\tilde{\maxaction}\right)\right],\\
    \approxbellman\qfunc(\state,\minaction,
    \maxaction)= \rwdfunc(\state,\minaction,
    \maxaction) + \df\ee_{\state'\sim\transitionkernel\left(\cdot\vert\state,\minaction,\maxaction\right)}\Big[\qfunc\big(\state', \oracle_\epsilon(\qfunc(\state',\cdot,\cdot))\big)\Big].
\end{align*}

Let us define the error at iteration $k$ as $\error_k:=\qfunckone - \bellman\qfunck$. As mentioned in \Cref{section: online minimax fqi}, besides the fact that we can only approximate $\bellman\qfunc_k$, there are additional sources of error: (i) the restriction of $\qfunckone\in\mathcal{F}$, while $\bellman\qfunck,\approxbellman\qfunck$ need not belong to $\mathcal{F}$, and (ii) a statistical error stemming from the use of finite samples. We decompose the error $\error_k$ as
\begin{align*}
        \error_k&=\underbrace{\left(\qfunckone-\approxbellman\qfunck\right)}_{\text{statistics + approximation}}+\underbrace{\left(\approxbellman\qfunck - \bellman\qfunck\right)}_{\text{oracle bias}}\\
        \implies \left\|\error_k\right\|_{2,\samplingdist} &\leq \left\|\qfunckone-\approxbellman\qfunck\right\|_{2,\samplingdist} + \left\|\approxbellman\qfunck-\bellman\qfunck\right\|_{2,\samplingdist}.
\end{align*}
From the use of $\oracle_\epsilon$ and \Cref{assumption: minimax oracle},
the oracle bias is uniformly bounded. It is not difficult to see that
\begin{align}\label{eq: oracle bias uniformly bounded}
    \left\|\approxbellman\qfunck-\bellman\qfunck\right\|_{2,\samplingdist} \leq \df\epsilon.
\end{align}
We will now analyze the $\|\qfunckone-\approxbellman\qfunck\|_{2,\samplingdist}$ term. Let $\xconc:=(\state,\minaction,\maxaction)\in\xspace:=(\statespace\times\actionspacemin\times\actionspacemax)$. It is clear from \Cref{algo: fqi} that $y_i$ is an unbiased conditional estimate of $\approxbellman\qfunck(\state_i,\minaction_i,\maxaction_i)=\approxbellman\qfunck(\xconc_i)$, i.e., $y_i=\approxbellman\qfunck + \ynoise_i$, where $\ee\left[\ynoise~\vert~\qfunck,\xconc_i\right] = 0$. We define 
\begin{align*}
    \bestfk\in\arg\min_{f\in\mathcal{F}}\left\|f-\approxbellman\qfunck\right\|_{2,\samplingdist}.
\end{align*}
It will be useful to quantify how close $\qfunckone$ is to $\approxbellman\qfunck$ empirically for the batch of data samples compared to $\bestfk$. We formalize this in the next lemma.
\begin{lemma}\label{lemma: separation bound}
    From \Cref{algo: fqi}, $\qfunckone = \arg\min_{f\in\mathcal F}\widehat L_n(f;\qfunck)$ yields
    \begin{align*}
        \frac{1}{n}\sum_{i=1}^n\left(\qfunckone(\xconc_i)-\approxbellman\qfunck(\xconc_i)\right)^2 \leq \frac{1}{n}\sum_{i=1}^n\left(\bar f(\xconc_i)-\approxbellman\qfunck(\xconc_i)\right)^2 + \frac{2}{n}\sum_{i=1}^n \xi_i\left(\qfunckone(\xconc_i)-\bar f(\xconc_i)\right).
    \end{align*}
\end{lemma}
\begin{proof}
    As $\qfunckone$ is optimal for the regression problem at iteration $k$, we have
    \begin{align*}
        \widehat L_n(\qfunckone;\qfunck) &\leq \widehat L_n(\bestfk;\qfunck)\\
        \implies \frac{1}{n}\sum_i\left(\qfunckone(\xconc_i)-y_i\right)^2 &\leq \frac{1}{n}\sum_i(\bar f(\xconc_i)-y_i)^2\\
        \implies \frac{1}{n}\sum_i\left(\qfunckone(\xconc_i)-\approxbellman\qfunck(\xconc_i)-\ynoise_i\right)^2 &\leq \frac{1}{n}\sum_i(\bar f(\xconc_i)-\approxbellman\qfunck(\xconc_i)-\ynoise_i)^2
    \end{align*}
    Opening squares, canceling the $\ynoise_i$ terms and rearranging yields the claim.
\end{proof}
The $\|\qfunckone-\approxbellman\qfunck\|_{2,\samplingdist}$ term can be bounded by uniformly bounding the two right hand side terms of \Cref{lemma: separation bound}. We will first bound the term $\frac{2}{n}\sum_{i=1}^n \xi_i\left(\qfunckone(\xconc_i)-\bar f(\xconc_i)\right)$. 

For $f\in\mathcal F$, let $Z_i(f):=\ynoise_i\big(f(\xconc_i)-\approxbellman\qfunck(\xconc_i)\big)$.
Note that we can apply Bernstein's inequality on $Z_i$. To see this, $\ee\left[ Z_i\right] = \ee_{\xconc_i}\left[ \ee\left[Z_i|\xconc_i\right]\right]=0$, and from \Cref{assumption: reward,assumption: FA}, we have $|y_i|\leq \qmax$, $|\approxbellman\qfunck(\xconc_i)|\leq\qmax$, which implies $|\ynoise_i|\leq 2\qmax$. Thus,
\begin{align*}
    |Z_i(f)|\leq 2\qmax\cdot\|f-\approxbellman\|_\infty\leq 4\qmax^2=:b
\end{align*}
Further, we have
\begin{align*}
    \mathrm{Var}(Z_i(f))\leq \ee\left[\xi_i^2(f(x_i)-\approxbellman\qfunck(x_i))^2\right]\le 4Q_{\max}^2\,\|f-\approxbellman\qfunck\|_{2,\sigma}^2 \le 16Q_{\max}^4=:\tau^2
\end{align*}
Thus, for $f\in\mathcal{F}$, for $\delta'\in(0,1)$, we have with probability at least $1-\delta'$ 
\begin{align*}
\frac{1}{n}\sum_i Z_i(f)\leq 4\qmax^2\sqrt{\frac{2\log(1/\delta')}{n}} + \frac{8\qmax^2\log(1/\delta')}{3n}
\end{align*}
To extend this statement uniformly to the whole family  $\mathcal{F}$, for a fixed $\eta_1>0$, let $\coveringnumber:=\coveringnumber(\eta_1):=\coveringnumber(\mathcal{F}, \eta_1, \|\cdot\|_\infty)$ be the minimal covering number of $\mathcal{F}$ in the sup norm with associated cover $\{f^{(1)},\dots,f^{(\coveringnumber)}\}$. Note that $f\mapsto Z_i(f)$ is $2\qmax$-Lipschitz in the sup norm. To see this, for $f,f'\in\mathcal{F}$, we have $|Z_i(f)-Z_i(f')|=|\ynoise\cdot\left(f(\xconc_i)-f'(\xconc_i)\right)|\leq 2\qmax \|f-f'\|_\infty$. Now, for each $f\in\mathcal{F}$, we have from the $\eta_1$-cover that $\exists~j\in[\coveringnumber]$ such that $\|f-f^{(j)}\|_\infty\leq \eta_1$. Thus,
\begin{align*}
    \sup_{f\in\mathcal{F}} \left|\frac{1}{n}\sum_i Z_i(f)\right| \leq \max_{1\leq j\leq \coveringnumber}\left|\frac{1}{n}\sum_i Z_i(f^{(j)})\right| +2\qmax\eta_1.
\end{align*}
We can use the point wise bounds on the cover points $f^{(j)}$ with $\delta'=\nicefrac{\delta}{2\coveringnumber(\eta_1)}$, where the factor of 2 comes because we bound the absolute value. Using the union bound yields with probability at least $1-\delta$
\begin{align}\label{eq: Z_i bernstein}
    \nonumber\sup_{f\in\mathcal F}\left|\frac1n\sum_i Z_i(f)\right| &\leq 4\qmax^2\sqrt{\frac{2\log(2\coveringnumber(\eta_1)/\delta)}{n}} + \frac{8\qmax^2\log(2\coveringnumber(\eta_1)/\delta)}{3n} + 2\qmax\eta_1 \\&=: \Delta_n(\eta_1,\delta).
\end{align}
Thus, we have
\begin{align}\label{eq: cross term bernstein}
\frac{2}{n}\sum_i\ynoise_i\left(\qfunckone(\xconc_i)-\bestfk(\xconc_i)\right) = \frac{2}{n}\sum_iZ_i(\qfunckone)- \frac{2}{n}\sum_iZ_i(\bestfk) \leq 4\Delta_n(\eta_1,\delta)
\end{align}

We need to also connect the empirical LHS term of \Cref{lemma: separation bound}, $ \frac{1}{n}\sum_{i=1}^n\big(\qfunckone(\xconc_i)-\approxbellman\qfunck(\xconc_i)\big)^2$ with the bound $\| \qfunckone - \approxbellman\qfunck\|^2_{2,\samplingdist}$. For $f\in\mathcal{F}$, let us define a centered random variable
\begin{align*}
    X_i:=\left(f(\xconc_i)-\approxbellman\qfunck(\xconc_i)\right)^2 - \left\| f - \approxbellman\qfunck\right\|^2_{2,\samplingdist}.
\end{align*}
To analyze $\frac{1}{n}\sum_{i=1}^n\big(\bar f(\xconc_i)-\approxbellman\qfunck(\xconc_i)\big)^2$, we can again apply Bernstein's inequality on for $X_i$. To see this, we have $\ee\big[X_i\big]=0$, $|X_i| = \big|(f(\xconc_i)-\approxbellman\qfunck(\xconc_i)\big)^2 - \big\| f - \approxbellman\qfunck\big\|^2_{2,\samplingdist}\big|\leq 4\qmax^2=b$, and $\variance(X_i) = \variance\big((f(\xconc_i)-\approxbellman\qfunck(\xconc_i))^2\big)\leq \ee\big[(f(\xconc_i)-\approxbellman\qfunck(\xconc_i))^4\big]\leq 16\qmax^4=\tau^2$. Also, the function $h(f):=(f(\xconc)-\approxbellman\qfunck(\xconc))^2$ is $4\qmax$-Lipschitz in the sup norm as for $f,f'\in\mathcal{F}$, $|h(f)-h(f')|=|f(\xconc)-f'(\xconc)|\cdot |f(\xconc)+f'(\xconc)-2\approxbellman\qfunck(\xconc)|\leq 4\qmax \|f-f'\|_\infty$. Analyzing similarly to $Z_i$ yields
\begin{align}\label{eq: X_i bernstein}
    \left|\frac{1}{n}\sum_i\left(f(\xconc_i)-\approxbellman\qfunck(\xconc_i)\right)^2 - \|f-\approxbellman\qfunck\|_{2,\sigma}^2\right| \leq \Delta_n(\eta_1,\delta)+6\qmax\eta_1.
\end{align}

We now analyze the term $\frac{1}{n}\sum_{i=1}^n\big(\bar f(\xconc_i)-\approxbellman\qfunck(\xconc_i)\big)^2$ from \Cref{lemma: separation bound}. Note that for each $i$, $\big(\bar f(\xconc_i)-\approxbellman\qfunck(\xconc_i)\big)^2$ is i.i.d. with mean $\ee_{\xconc\sim\samplingdist}\big[\big(\bar f(\xconc)-\approxbellman\qfunck(\xconc)\big)^2\big]=\big\|\big(\bar f-\approxbellman\qfunck\big)\big\|^2_{2,\samplingdist}\leq \delapprox^2$ from the definition of $\delapprox$ from \Cref{eq: delapprox}. Also $\big(\bar f(\xconc_i)-\approxbellman\qfunck(\xconc_i)\big)^2\in\left[0,4\qmax^2\right]$. Thus, we can make use of Hoeffding's inequality. For $\delta>0$, we have with probability at least $1-\delta$ that
\begin{align}\label{eq: f bar hoeffding}
    \nonumber\frac{1}{n}\sum_i\left(\bestfk(\xconc_i)-\approxbellman\qfunck\right)^2&\leq \left\|\left(\bar f-\approxbellman\qfunck\right)\right\|^2_{2,\samplingdist} + 4\qmax^2\sqrt{\frac{\log\left(\nicefrac{1}{\delta}\right)}{2n}}\\
    &\leq \delapprox^2+ 4\qmax^2\sqrt{\frac{\log\left(\nicefrac{1}{\delta}\right)}{2n}}.
\end{align}

We will now combine these inequalities to get a per iteration error bound on $\error_k$. 
\begin{lemma}[Per iteration error bound]\label{lemma: per iteration error bound}
    Under \Cref{assumption: actions,assumption: reward,assumption: dyn,assumption: exploratory dist,assumption: FA,assumption: minimax oracle}, for $\delta>0$, for $\error_k=\|\qfunckone-\bellman\qfunck\|$, we have with probability at least $\geq 1-\delta$, the error at the $k^{th}$ iteration
    \begin{align*}
        \|\error_k\|_{2,\samplingdist} \leq \gamma\epsilon + \sqrt{\delapprox^2+ 5\Delta_n(\eta_1,\frac{\delta}{3}) + 4\qmax^2\sqrt{\frac{\log\left(\nicefrac{3}{\delta}\right)}{2n}} + 6\qmax\eta_1}~.
    \end{align*}
\end{lemma}
\begin{proof}
    We have
    \begin{align*}
    \|\qfunckone-\approxbellman\qfunck\|_{2,\sigma}^2&\stackrel{\Cref{eq: X_i bernstein}}{\leq} \frac{1}{n}\sum_i\left(\qfunckone\left(\xconc_i\right)-\approxbellman\qfunck\left(\xconc_i\right)\right)^2+\Delta_n(\eta_1,\delta)+6\qmax\eta_1\\
    &\stackrel{\Cref{lemma: separation bound}}{\leq} \frac{1}{n}\sum_{i=1}^n\left(\bar f(\xconc_i)-\approxbellman\qfunck(\xconc_i)\right)^2 + \frac{2}{n}\sum_{i=1}^n \xi_i\left(\qfunckone(\xconc_i)-\bar f(\xconc_i)\right)\\
    &\quad\quad\quad\quad\quad\quad+\Delta_n(\eta_1,\delta)+6\qmax\eta_1\\
    &\stackrel{\Cref{eq: f bar hoeffding}}{\leq}\frac{2}{n}\sum_{i=1}^n \xi_i\left(\qfunckone(\xconc_i)-\bar f(\xconc_i)\right)+ \delapprox^2+ 4\qmax^2\sqrt{\frac{\log\left(\nicefrac{1}{\delta}\right)}{2n}}\\
    &\quad\quad\quad\quad\quad\quad+\Delta_n(\eta_1,\delta)+6\qmax\eta_1\\
    &\stackrel{\Cref{eq: cross term bernstein}}{\leq} \delapprox^2+ 5\Delta_n(\eta_1,\delta) + 4\qmax^2\sqrt{\frac{\log\left(\nicefrac{1}{\delta}\right)}{2n}} + 6\qmax\eta_1.
\end{align*}
\end{proof}

\subsection{Error Propogation Across Iterations}\label{subsection: error prop across iterations}
We have so far a bound on $\|\error_k(\state,\minaction,\maxaction)\|_{2,\samplingdist}$. For error propogation across iterations, we shall need a bound on $\bar{\error}_k:=\|\sup_{\minaction,\maxaction}\error_k(\state,\minaction,\maxaction)\|_{2,\samplingdist}$
To that end, we have the following result for our setting, analogous to single-agent versions used in prior work \citep{antos2007fitted}.
\begin{lemma}\label{lemma: 2 norm to sup}
    Let $f:\statespace
    \times\actionspacemin\times\actionspacemax\rightarrow\bb{R}$ be such that $\forall~\state\in\statespace,f(\state,\cdot,\cdot)$ is jointly $L_f$-Lipschitz and non-negative. Let $\mathcal{B}(\bm a, \eta)$ denote an $\ell^2$ ball in $\actionspacemin\times\actionspacemax$ centered at $\bm a\in\actionspacemin\times\actionspacemax$ with radius $\eta$. Then for $\eta_2>0$, we have
    \begin{align*}
        \ee_{\state\sim\samplingdist_\statespace}\left[\sup_{\minaction,\maxaction}f(\state,\minaction,\maxaction)^2\right] \leq \frac{4}{c\volume\left(\nicefrac{\eta_2}{2}\right)}\|f(\state,\minaction,\maxaction)\|_{2,\sigma}^2 + 3L_f^2\eta_2^2,
    \end{align*}
where $\volume(\eta)$ is the lebesgue measure $V(\eta)=\inf_{\bm a\in\actionspacemin\times\actionspacemax}\lambda\Big(\mathcal{B}(\bm a, \eta)\cap (\actionspacemin\times\actionspacemax)\Big)$.
\end{lemma}
\begin{proof}
    Let $\{(\minaction^j,\maxaction^j)\}_{j\in[\coveringnumber(\eta_2)]}$ be the minimal cover of $\actionspacemin\times\actionspacemax$ associated with $\eta_2>0$ in the $\|\cdot\|_2$ norm. Consider $(\minaction,\maxaction)$ and let its closest cover point be $(\minaction^j,\maxaction^j)$. Then we have
    \begin{align} \label{eq: propgate cover lower bound}
       \nonumber &\left\vert f(\state,\minaction,\maxaction)-f(\state,\minaction^j,\maxaction^j)\right\vert \leq L_f\eta_2\\
       \nonumber \implies &\left\vert f(\state,\minaction,\maxaction)\right\vert \leq \left|f(\state,\minaction^j,\maxaction^j)\right| + L_f\eta_2\\
        \implies &\sup_{\minaction,\maxaction}f(\state,\minaction,\maxaction)^2 \leq 2\max_jf(\state,\minaction^j,\maxaction^j)^2 + 2L_f^2\eta_2^2
    \end{align}
    Now we bound $\max_jf(\state,\minaction^j,\maxaction^j)^2$. We have for $(\minaction,\maxaction)$ in a radius $\nicefrac{\eta_2}{2}$ around $(\minaction^j,\maxaction^j)$, i.e., $(\minaction,\maxaction)\in\mathcal{B}\left((\minaction^j,\maxaction^j), \nicefrac{\eta_2}{2}\right)$, that
    \begin{align*}
        f(\state,\minaction,\maxaction) - f(\state,\minaction^j,\maxaction^j) &\geq -\frac{L_f\eta_2}{2}\\
        \implies f(\state,\minaction,\maxaction)^2&\geq \left(\max\left\{f(\state,\minaction^j,\maxaction^j) - \frac{L_f\eta_2}{2},0\right\}\right)^2\\
        &\geq \frac{f(\state,\minaction^j,\maxaction^j)^2}{2} - \frac{L_f^2\eta_2^2}{4}.
    \end{align*}
    Here we used that $\left(\max\{a-b,0\}\right)^2\geq \nicefrac{a^2}{2}-b^2$ for positive $a,b$.
Now, $\|f(\state,\minaction,\maxaction)\|_{2,p(\cdot,\cdot|s)}^2 \geq \int_{\mathcal{B}\left((\minaction^j,\maxaction^j), \nicefrac{\eta_2}{2}\right)}f(\state,\minaction,\maxaction)^2p(\minaction,\maxaction|\state)d(\minaction,\maxaction)$, meaning that $\exists ~c>0$ such that
\begin{align}\label{eq: propgate cover upper bound}
    &\nonumber\|f(\state,\minaction,\maxaction)\|_{2,p(\cdot,\cdot|s)}^2 \\
    &\geq \int_{\mathcal{B}\left((\minaction^j,\maxaction^j), \nicefrac{\eta_2}{2}\right)}f(\state,\minaction,\maxaction)^2p(\minaction,\maxaction|\state)d(\minaction,\maxaction)\geq c\volume\left(\frac{\eta_2}{2}\right) \left(\frac{f(\state,\minaction^j,\maxaction^j)^2}{2} - \frac{L_f^2\eta_2^2}{4}\right)\\
    &\implies f(\state,\minaction^j,\maxaction^j)^2 \leq \frac{2}{c\volume\left(\nicefrac{\eta_2}{2}\right)}\|f(\state,\minaction,\maxaction)\|_{2,p(\cdot,\cdot|s)}^2  + \frac{L^2_f\eta_2^2}{2}.
\end{align}
Combining \Cref{eq: propgate cover lower bound,eq: propgate cover upper bound}, we get that
\begin{align*}
\sup_{\minaction,\maxaction}f(\state,\minaction,\maxaction)^2 \leq \frac{4}{c\volume\left(\nicefrac{\eta_2}{2}\right)}\|f(\state,\minaction,\maxaction)\|_{2,p(\cdot,\cdot|s)}^2 + 3L_f^2\eta_2^2\\
\implies \ee_{\state\sim\samplingdist_\statespace}\left[\sup_{\minaction,\maxaction}f(\state,\minaction,\maxaction)^2\right] \leq \frac{4}{c\volume\left(\nicefrac{\eta_2}{2}\right)}\|f(\state,\minaction,\maxaction)\|_{2,\sigma}^2 + 3L_f^2\eta_2^2.
\end{align*}
This completes the proof.
\end{proof}
We shall use \Cref{lemma: 2 norm to sup} to analyze error propagation across iterations. We define $\dk:=\qfunc^*-\qfunck$. Note that $\dkone=\qfunc^*-\qfunckone = \bellman\qfunc^*-\bellman\qfunckone-\error_k$. Thus
\begin{align*}
    \dkone(\state,\minaction,\maxaction) = \left(\bellman\qfunc^*-\bellman\qfunck\right)(\state,\minaction,\maxaction) - \error_k(\state,\minaction,\maxaction).
\end{align*}
Let $e_k(\state):=\sup_{\minaction,\maxaction}\left|\dk(\state,\minaction,\maxaction)\right|$. Then, we have
\begin{align*}
    \forall~s\in\statespace,~e_{k+1}(\state)\leq \df \sup_{\minaction,\maxaction}\ee_{\state'\sim\transitionkernel(\cdot|\state,\minaction,\maxaction)}\left[e_k(\state')\right] + \sup_{\minaction,\maxaction}\left|\error_k(\state,\minaction,\maxaction)\right|.
\end{align*}
At step $K$, for $\st{0}\sim\samplingdist_\statespace$, we have
\begin{align*}
    &e_{K}(\st{0})\leq \df \sup_{\ut{0},\vt{0}}\ee_{\st{1}\sim\transitionkernel(\cdot|\state,\ut{0},\vt{0})}\left[e_{K-1}(\st{1})\right] + \sup_{\minaction,\maxaction}\left|\error_{K-1}(\st{0},\minaction,\maxaction)\right|\\
    &\leq \df \sup_{\ut{0},\vt{0}}\ee_{\st{1}}\left[\df\sup_{\ut{1},\vt{1}}\ee_{\st{2}\sim\transitionkernel(\cdot|\st{1},\ut{1},\vt{1})}\left[e_{K-2}(\st{2})\right] + \sup_{\minaction,\maxaction}\left|\error_{K-2}(\st{1},\minaction,\maxaction)\right|\right]\\
    &\quad\quad\quad\quad\quad\quad\quad\quad\quad\quad\quad\quad\quad\quad\quad+ \sup_{\minaction,\maxaction}\left|\error_{K-1}(\st{0},\minaction,\maxaction)\right|\\
    &\leq \df^2 \sup_{\ut{0},\vt{0}}\ee_{\st{1}}\left[\sup_{\ut{1},\vt{1}}\ee_{\st{2}}\left[e_{K-2}(\st{2})\right]\right] + \df\sup_{\ut{0},\vt{0}}\ee_{\st{1}}\left[\sup_{\minaction,\maxaction}\left|\error_{K-2}(\st{1},\minaction,\maxaction)\right|\right]\\
&\quad\quad\quad\quad\quad\quad\quad\quad\quad\quad\quad\quad\quad\quad\quad+ \sup_{\minaction,\maxaction}\left|\error_{K-1}(\st{0},\minaction,\maxaction)\right|\\
    &\qquad\qquad\qquad\qquad\qquad\qquad\qquad\qquad\qquad\qquad\vdots\\
    &\leq \df^K \sup_{\substack{\ut{0},\\ \vt{0}}}\ee_{\st{1}}\Bigg[\sup_{\substack{\ut{1},\\ \vt{1}}}\ee_{\st{2}}\Big[\cdots \sup_{\substack{\ut{K-1},\\ \vt{K-1}}} \ee_{\st{K}}\left[e_{0}(\st{K})\right]\Big]\Bigg]\\
    &\quad\quad\quad\quad+ \sum_{k=0}^{K-2}\df^{K-1-k} \sup_{\substack{\ut{0},\\ \vt{0}}}\ee_{\st{1}}\left[\cdots\sup_{\substack{\ut{K-2-k},\\ \vt{K-2-k}}}\ee_{\st{K-1-k}}\left[\sup_{\minaction,\maxaction}\left|\error_{k}(\st{K-1-k},\minaction,\maxaction)\right|\right]\right]\\
    &\quad\quad\quad\quad+\sup_{\minaction,\maxaction}\left|\error_{K-1}\left(\state_{(0)},\minaction,\maxaction\right)\right|
\end{align*}

Consider the term $e_0$. We have $\forall~\state\in\statespace, e_0(\state)\leq \|D_k\|_\infty\leq 2\qmax$. For the second term, we have from \Cref{assumption: discounted conc coeff} that
\begin{align*}
    &\ee_{\st{0}\sim\samplingdist_\statespace}\left[ \sup_{\substack{\ut{0},\\ \vt{0}}}\ee_{\st{1}}\left[\cdots\sup_{\substack{\ut{K-2-k},\\ \vt{K-2-k}}}\ee_{\st{K-1-k}}\left[\sup_{\minaction,\maxaction}\left|\error_{k}(\st{K-1-k},\minaction,\maxaction)\right|\right]\right]\right] \\&
    \leq c_{K-1-k} \ee_{\state\sim\samplingdist_\statespace}\left[\sup_{\minaction,\maxaction}\left|\error_{k}(\state,\minaction,\maxaction)\right|\right]\\
    &\leq c_{K-1-k} \sqrt{\ee_{\state\sim\samplingdist_\statespace}\left[\sup_{\minaction,\maxaction}\error_{k}(\state,\minaction,\maxaction)^2\right]}, \tag{Jensen's inequality $(\ee[X])^2\leq \ee[X^2]$} = c_{K-1-k}\bar{\error}_k
\end{align*}
where we write $c_j:=c_j(\samplingdist_\statespace)$. Thus, we have
\begin{align*}
    \ee_{\state\sim\samplingdist}\left[e_K(\state)\right] \leq \df^K\cdot2\qmax + \sum_{k=0}^{K-1}\df^{K-1-k}c_{K-1-k}\bar{\error}_k
\end{align*}

Now we have for $\state\in\statespace$, $\minaction,\minaction'\in\actionspacemin$ and $\maxaction,\maxaction'\in\actionspacemax$
\begin{align*}
    \left|\error_k(\state,\minaction,\maxaction)-\error_k(\state,\minaction',\maxaction')\right| &\leq \left|\qfunck(\state,\minaction',\maxaction') - \qfunck(\state,\minaction,\maxaction) \right| + \left|\bellman\qfunck(\state,\minaction',\maxaction') - \bellman\qfunck(\state,\minaction,\maxaction) \right|\\
    &\leq L_\mathcal{F}\|(\minaction,\maxaction)-(\minaction',\maxaction')\|_2 +  \left|\rwdfunc(\state,\minaction',\maxaction') - \rwdfunc(\state,\minaction,\maxaction) \right| \\
    &\quad+ \df \left|\int\min_{\tilde \minaction\in\actionspacemin}\max_{\tilde\maxaction\in\actionspacemax}\qfunck(\state',\tilde\minaction,\tilde\maxaction)\Big(\transitionkernel(d\state'|\minaction,\maxaction)-\transitionkernel(d\state'|\minaction',\maxaction')\Big)\right|\\
    &\leq \underbrace{\sqrt{2}\left(L_\mathcal{F}+L_\rwdfunc + 2\df\qmax L_\transitionkernel\right)}_{:=L_\error}\|(\minaction,\maxaction)-(\minaction',\maxaction')\|_2,
\end{align*}
implying that $\error_k$ is $L_\error:=\sqrt{2}\left(L_\mathcal{F}+L_\rwdfunc + 2\df\qmax L_\transitionkernel\right)$-Lipschitz in the joint actions. Thus from \Cref{lemma: 2 norm to sup}, we have for $\bar{\error}_k=\|\sup_{\minaction,\maxaction}|\error_k(\state,\minaction,\maxaction)|\|_{2,\samplingdist}$ for $\eta_2>0$
\begin{align*}
    \bar{\error}_k^2 &\leq \frac{4}{c\volume\left(\nicefrac{\eta_2}{2}\right)}\|\error_k\|_{2,\samplingdist}^2 + 3L_\error^2\eta_2^2\\
    \implies \bar{\error}_k &\leq \frac{2}{\sqrt{c\volume\left(\nicefrac{\eta_2}{2}\right)}}\|\error_k\|_{2\samplingdist} + \sqrt{3}L_\error\eta_2.
\end{align*}
Thus, for $\delta>0$, from \Cref{lemma: per iteration error bound}, using the union bound for $k$ iterations, we have with probability at least $1-\delta$ 
\begin{align*}
    &\forall~k\in[\numiters],\bar{\error}_k\leq \\
    &
    \underbrace{\sqrt{3}L_\error\eta_2 +\frac{2}{\sqrt{c\volume\left(\nicefrac{\eta_2}{2}\right)}} \left(\gamma\epsilon + \sqrt{\delapprox^2+ 5\Delta_n\left(\eta_1,\frac{\delta}{3\numiters}\right) + 4\qmax^2\sqrt{\frac{\log\left(\nicefrac{3\numiters}{\delta}\right)}{2n}} + 6\qmax\eta_1}\right)}_{:=b}
\end{align*}
This yields that with probability at least $1-\delta$
\begin{align*}
    \ee_{\state\sim\samplingdist}\left[e_\numiters(\state)\right] &\leq \df^\numiters\cdot2\qmax + \sum_{k=0}^{\numiters-1}\df^{\numiters-1-k}c_{\numiters-1-k} \cdot b\\
    &\leq \df^\numiters\cdot2\qmax + \frac{C_{\samplingdist_{\statespace}}}{1-\df} b\\
    \implies \ee_{\state\sim\rho}\left[e_\numiters(\state)\right] &\leq \df^\numiters\cdot2\qmax C_{\rho,\samplingdist} + C_{\samplingdist_{\statespace}} \frac{C_{\rho,\samplingdist}C_{\samplingdist_{\statespace}}}{1-\df} b, \tag{$\nicefrac{d\rho}{d\samplingdist}\leq C_{\rho,\samplingdist}$}
\end{align*}
where in the last step we used $\ee_\rho[g] \leq C_{\rho,\samplingdist}\ee_\samplingdist[g]$ for non-negative $g$ by the definition of $C_{\rho,\samplingdist}$ from the statement of \Cref{theorem: fqi bounds}. This completes the proof for \Cref{theorem: fqi bounds}.

\section{Proof of \Cref{corollary: nash gap}}\label{appendix: nash gap proof}
    Let us denote the true security strategies corresponding to $\qfunc\in\mathcal{F}$ as $\pi^*_{\qfunc}(\state)\in\arg\min_{\minaction\in\actionspacemin}\max_{\maxaction\in\actionspacemax}\qfunc(\state,\minaction,\maxaction)$ and $\mu^*_{\qfunc}(\state)\in\arg\max_{\maxaction\in\actionspacemax}\min_{\minaction\in\actionspacemin}\qfunc(\state,\minaction,\maxaction)$. We will first bound the entity $\max_{\mu'}\valfunc^{\pi_{\qfunc_\numiters},\mu'}-\valfunc^*$. Consider any state $\state$ and Markov policy $\mu'$. Let $e_K(\state):=\sup_{\minaction\in\actionspacemin,\maxaction\in\actionspacemax}\left|\qfunc_\numiters(\state,\minaction,\maxaction)-\qfunc^*(\state,\minaction,\maxaction)\right|$. 
    Now consider the difference $\tilde{\Delta}(\state):=\valfunc^{\pi_{\qfunc_\numiters},\mu'}(\state)-\valfunc^*(\state)=\valfunc^{\pi_{\qfunc_\numiters},\mu'}(\state)-\qfunc^*(\state,\pi_{\qfunc_\numiters},\mu'(\state))+\qfunc^*(\state,\pi_{\qfunc_\numiters},\mu'(\state))-\valfunc^*(\state)$. From the definition of $\valfunc$ and $\qfunc^*$ we have
    \begin{align*}
        \valfunc^{\pi_{\qfunc_\numiters},\mu'}(\state)-\qfunc^*(\state,\pi_{\qfunc_\numiters},\mu'(\state)) &= \df \ee_{\state_1\sim\transitionkernel(\cdot|\state,\pi_{\qfunc_\numiters}(\state),\mu'(\state))}\left[\valfunc^{\pi_{\qfunc_\numiters},\mu'}(\state_1) - \valfunc^*(\state_1)\right]\\
        &=\df \ee_{\state_1\sim\transitionkernel(\cdot|\state,\pi_{\qfunc_\numiters}(\state),\mu'(\state))}\left[\tilde\Delta(\state_1)\right]\\
    \end{align*}
    Thus we have the recursion
    \begin{align*}
        \tilde{\Delta}(\state) = \qfunc^*(\state,\pi_{\qfunc_\numiters},\mu'(\state))-\valfunc^*(\state) + \df \ee_{\state_1\sim\transitionkernel(\cdot|\state,\pi_{\qfunc_\numiters}(\state),\mu'(\state))}\left[\tilde\Delta(\state_1)\right].
    \end{align*}
    This yields
    \begin{align}\label{eq: nash gap first half}
        \valfunc^{\pi_{\qfunc_\numiters},\mu'}(\state)-\valfunc^*(\state)&=\sum_{t=0}^\infty\df^t\ee_{\state_t|(\state_0=\state,\pi_{\qfunc_\numiters},\mu')}\left[\qfunc^*(\state_t,\pi_{\qfunc_\numiters}(\state_t),\mu'(\state_t))-\valfunc^*(\state_t)\right],
    \end{align}
    where we use $\ee_{\state_t|(\state_0=\state,\pi_{\qfunc_\numiters},\mu')}$ as shorthand to denote the iterated expectation over the states from $\state_0$ to $\state_t$ under player policies $\pi_{\qfunc_\numiters}$ and $\mu'$. To bound $\qfunc^*(\state,\pi_{\qfunc_\numiters}(\state),\mu'(\state))-\valfunc^*(\state)$.
    We have that 
    \begin{align*}
        \qfunc^*(\state,\pi_{\qfunc_\numiters}(\state),\mu'(\state)) &\leq \qfunc^*(\state,\pi_{\qfunc_\numiters},\tilde \maxaction),~\tilde \maxaction\in\arg\max_{\maxaction\in\actionspacemax}\qfunc^*(\state,\pi_{\qfunc_\numiters},\maxaction)\\
        &\leq \qfunc_\numiters(\state,\pi_{\qfunc_\numiters},\tilde \maxaction) + e_K(\state)\\
        &\leq \max_{\maxaction\in\actionspacemax}\qfunc_\numiters(\state,\pi_{\qfunc_\numiters},\maxaction) +e_K(\state)\\
        &\leq \min_{\minaction\in\actionspacemin}\qfunc_\numiters(\state,\minaction,\mu_{\qfunc_\numiters}) +e_K(\state) + \epsilon\\
        &\leq \qfunc_\numiters(\state,\tilde \minaction,\mu_{\qfunc_\numiters}) +e_K(\state) + \epsilon, \tilde\minaction\in\arg\min_{\minaction\in\actionspacemin}\qfunc^*(\state,\minaction, \mu_{\qfunc_\numiters})\\
        &\leq \qfunc^*(\state,\tilde \minaction,\mu_{\qfunc_\numiters}) +2e_K(\state) + \epsilon\\
        &= \min_{\minaction\in\actionspacemin}\qfunc^*(\state,\minaction,\mu_{\qfunc_\numiters}) +2e_K(\state) + \epsilon\\
        &\leq \underbrace{\max_{\maxaction\in\actionspacemax}\min_{\minaction\in\actionspacemin}\qfunc^*(\state,\minaction,\maxaction)}_{=\valfunc^*(\state)} +2e_K(\state) + \epsilon
    \end{align*}
    Using this in \Cref{eq: nash gap first half}, we get
    \begin{align*}
        \valfunc^{\pi_{\qfunc_\numiters},\mu'}(\state)-\valfunc^*(\state)\leq\sum_{t=0}^\infty\df^t\ee_{\state_t|(\state_0=\state,\pi_{\qfunc_\numiters},\mu')}\left[2e_K(\state_t)+\epsilon\right],
    \end{align*}
    From \Cref{assumption: discounted conc coeff}, we have $\ee_{\state\sim\rho}\left[e_K(\state_t)\right] \leq c_t(\rho) \ee_{\state\sim\samplingdist_\statespace}\left[e_K(\state_t)\right]$. Thus
    \begin{align*}
       \sum_{t=0}^\infty\df^t\ee_{\state\sim\rho}\left[2e_K(\state_t)+\epsilon\right] &\leq \sum_{t=0}^{\infty}\df^t\left(2c_t(\rho)\ee_{\state\sim\sigma_\statespace}\left[e_K(\state)\right]+ \epsilon\right)\\
       &\leq \frac{2C_\rho}{1-\df}\ee_{\state\sim\sigma_\statespace}\left[e_K(\state)\right] +\frac{\epsilon}{1-\df}.
    \end{align*}
Thus we have
\begin{align}
   \nonumber \ee_{\state\sim\rho}\left[\valfunc^{\pi_{\qfunc_\numiters},\mu'}(\state)-\valfunc^*(\state)\right]\leq \frac{2C_\rho}{1-\df}\ee_{\state\sim\sigma_\statespace}\left[e_K(\state)\right] +\frac{\epsilon}{1-\df}\\
    \implies\ee_{\state\sim\rho}\left[\max_{\mu'}\valfunc^{\pi_{\qfunc_\numiters},\mu'}(\state)-\valfunc^*(\state)\right]\leq \frac{2C_\rho}{1-\df}\ee_{\state\sim\sigma_\statespace}\left[e_K(\state)\right] +\frac{\epsilon}{1-\df}. \label{eq: nash gap term 1}
\end{align}
Similarly, for a policy $\pi'$, let us define $\tilde{\Delta}'(\state):=\valfunc^*(\state)-\valfunc^{\pi',\mu_{\qfunc_\numiters}}(\state)=\valfunc^*(\state)-\qfunc^*(\state,\pi',\mu_{\qfunc_\numiters})+\qfunc^*(\state,\pi',\mu_{\qfunc_\numiters}) -\valfunc^{\pi',\mu_{\qfunc_\numiters}}(\state)$. From the definition of $\valfunc^*,\qfunc^*$,
\begin{align*}
    \qfunc^*(\state,\pi',\mu_{\qfunc_\numiters}) -\valfunc^{\pi',\mu_{\qfunc_\numiters}}(\state) &= \df \ee_{\state_1\sim\transitionkernel(\cdot|\state,\pi_{\qfunc_\numiters},\mu')}\left[\valfunc^*(\state_1)-\valfunc^{\pi',\mu_{\qfunc_\numiters}}(\state_1)\right]\\
    &=\df \ee_{\state_1\sim\transitionkernel(\cdot|\state,\pi_{\qfunc_\numiters},\mu')}\left[\tilde \Delta '(\state_1)\right].
\end{align*}
This gives the recursion
\begin{align*}
    \tilde{\Delta}'(\state) = \valfunc^*(\state)-\qfunc^*(\state,\pi',\mu_{\qfunc_\numiters}) + \df \ee_{\state_1\sim\transitionkernel(\cdot|\state,\pi_{\qfunc_\numiters},\mu')}\left[\tilde \Delta '(\state_1)\right],
\end{align*}
which yields
\begin{align}\label{eq: nash gap second half}
    \valfunc^*(\state)-\valfunc^{\pi',\mu_{\qfunc_\numiters}}(\state)=\sum_{t=0}^\infty\df^t\ee_{\state_t|(\state_0=\state,\pi_{\qfunc_\numiters},\mu')}\left[\valfunc^*(\state_t)-\qfunc^*(\state_t,\pi'(\state_t),\mu_{\qfunc_\numiters}(\state_t))\right]
\end{align}
To bound $\valfunc^*(\state)-\qfunc^*(\state,\pi'(\state),\mu_{\qfunc_\numiters}(\state))$, similar to before we have
\begin{align*}
    \qfunc^*(\state,\pi'(\state),\mu_{\qfunc_\numiters}(\state)) &\geq \qfunc^*(\state,\tilde \minaction,\mu_{\qfunc_\numiters}(\state)), \tilde\minaction\in\arg\min_{\minaction}\qfunc^*(\state,\minaction,\mu_{\qfunc_\numiters}(\state))\\
    &\geq \qfunc_\numiters(\state,\tilde \minaction,\mu_{\qfunc_\numiters}(\state))-e_K(\state)\\
    &\geq \min_{\minaction\in\actionspacemin}\qfunc_\numiters(\state, \minaction,\mu_{\qfunc_\numiters}(\state))-e_K(\state)\\
    &\geq \max_{\maxaction\in\actionspacemax}\qfunc_\numiters(\state,\pi_{\qfunc_\numiters},\maxaction) - e_K(\state)-\epsilon\\
    &\geq \qfunc_\numiters(\state,\pi_{\qfunc_\numiters},\tilde \maxaction) - e_K(\state)-\epsilon, ~\tilde\maxaction\in\arg\max_{\maxaction\in\actionspacemax}\qfunc^*(\state,\pi_{\qfunc_\numiters}(\state),\maxaction)\\
    &\geq \qfunc^*(\state,\pi_{\qfunc_\numiters},\tilde \maxaction) - 2e_K(\state)-\epsilon\\
    &= \max_{\maxaction\in\actionspacemax}\qfunc^*(\state,\pi_{\qfunc_\numiters},\maxaction) - 2e_K(\state)-\epsilon\\
    &\geq \underbrace{\min_{\minaction\in\actionspacemin}\max_{\maxaction\in\actionspacemax}\qfunc^*(\state,\minaction,\maxaction)}_{:=\valfunc^*(\state)} - 2e_K(\state)-\epsilon\\
    \implies \valfunc^*(\state) &\leq \qfunc^*(\state,\pi'(\state),\mu_{\qfunc_\numiters}(\state)) + 2e_K(\state)+\epsilon.
\end{align*}
Using in \Cref{eq: nash gap second half}, we get
\begin{align*}
     \valfunc^*(\state)-\valfunc^{\pi',\mu_{\qfunc_\numiters}}(\state)\leq\sum_{t=0}^\infty\df^t\ee_{\state_t|(\state_0=\state,\pi_{\qfunc_\numiters},\mu')}\left[2e_K(\state_t)+\epsilon\right].
\end{align*}
From similar analysis as previously, this yields for $\state\sim\rho$
\begin{align}
    \ee_{\state\sim\rho}\left[\valfunc^*(\state)-\min_{\pi'}\valfunc^{\pi',\mu_{\qfunc_\numiters}}(\state)\right]\leq\frac{2C_\rho}{1-\df}\ee_{\state\sim\sigma_\statespace}\left[e_K(\state)\right] +\frac{\epsilon}{1-\df} \label{eq: nash gap term two}
\end{align}
Combining \Cref{eq: nash gap term 1,eq: nash gap term two}, we get for  $\state\sim\rho$
\begin{align*}
  \ee_{\state\sim\rho} \left[ \max_{\mu'}\valfunc^{\pi_{\qfunc_\numiters},\mu'}(\state)-\min_{\pi'}\valfunc^{\pi',\mu_{\qfunc_\numiters}}(\state)\right]&\leq\frac{4C_\rho}{1-\df}\ee_{\state\sim\sigma_\statespace}\left[e_K(\state)\right] +\frac{2\epsilon}{1-\df}\\
    \implies \texttt{NashGap}_\rho(\pi_{\qfunc_\numiters},\mu_{\qfunc_\numiters})&\leq\frac{4C_\rho}{1-\df}\ee_{\state\sim\sigma_\statespace}\left[\sup_{\minaction\in\actionspacemin,\maxaction\in\actionspacemin}|\qfunc_K-\qfunc^*|\right] +\frac{2\epsilon}{1-\df}
\end{align*}

\section{On Partially Input Convex-Concave Neural Network Architecture} \label{appendix: convex concave fa piccnn}

In this section, we describe the architecture underlying the function class $\mathcal{F}$ of
Section~\ref{section: convex concave fa}. Throughout, we use the softplus action function $\softmax(x) := \log(1+e^{x})$ applied entrywise to vectors and matrices. Softplus satisfies the PICNN requirement of being convex and non-decreasing in the input required for guaranteeing desirable convexity in the partial input \citep{amos2017input}. We first review the PICNN architecture.

\paragraph{Partially Input Convex Neural Networks \citep{amos2017input}} Fix $i \in \{1,2\}$ and write $(\x,\y) = (\state,\minaction)$ for $i=1$ and $(\x,\y)=(\state,\maxaction)$ for
$i=2$. For convenience, we assume both blocks $N^i_{\params^i}$, $i\in\{1,2\}$ have hidden width $m$ for the convex path corresponding to $\y$ and $p$ for the state path corresponding to $\x$, and depth $L$. All results presented here follow for different, network specific choices of $m,p$ as well. The architecture involves two recursions, one related with the \emph{state} path, denoted by $ h^i$, and the other related with the \emph{convex} path, denoted by $ z^i$, give by
\begin{equation}\label{eq:picnn}
    \begin{aligned}
        h^i_0 := \bm x, \qquad
 h^i_{\ell+1}& = \varphi\Big(\widetilde W^i_\ell\, h^i_\ell + \widetilde b^i_\ell\Big),
\\
z^i_0=\bm 0,\quad z^i_{\ell+1} = \varphi\Bigg(
  W^{(z),i}_\ell \Big(z^i_\ell \odot &\varphi\big(W^{(zh),i}_\ell h^i_\ell + b^{(z),i}_\ell\big)\Big)
  ~+\\
  W^{(y),i}_\ell \Big(\bm y \odot& \big(W^{(yh),i}_\ell\, h^i_\ell + b^{(y),i}_\ell\big)\Big)
  +
  W^{(h),i}_\ell\, h^i_\ell + b^i_\ell
\Bigg),\\
N^i(\x,\y)=(w^i)^\top z^i_L, &\quad\text{where}~\ell=0,\dots,L-1.
    \end{aligned}
\end{equation}
Convexity of $y \mapsto N^i_{\theta_i}(\x,\y)$ requires $W^{(z),i}_\ell$ and $w^i$ to
have non-negative entries \citep{amos2017input}. This can be done in multiple ways, for example, by projecting the $W^i_\ell$ and $w^i$ parameters to the non-negative set while training. In our work, to ensure smoothness, we enforce non-negativity through unconstrained reparameterization where one can optimize over free parameters
$\tilde {W}^{(z),i}_\ell \in \mathbb{R}^{m\times m}$ and
$\widetilde w^i \in \mathbb{R}^{m}$ and set
\begin{equation}\label{eq:reparam}
  {W}^{(z),i}_\ell = \varphi\big(\tilde {W}^{(zh),i}_\ell\big),
  \qquad
  w^i = \varphi\big(\widetilde w^i\big)
\end{equation}
\Cref{eq:picnn} implies that the network $N^i$ has
\begin{align}
    \begin{aligned}
    |\params^i|:=&L(m^2+md_{\y}+2m+d_{\y})+m+pd_{\state}+(p^2+2mp+pd_{\y})(L-1)+pL\\
    & +2md_{\state}+d_{\y}d_{\state}
    \end{aligned}
     \label{eq: param count picnn}
\end{align}
parameters, with $d_{\y}=d_{\minaction}$ for $i=1$ and $d_{\y}=d_{\maxaction}$ for $i=2$.
We use this partially input convex $N^i$ to construct our $\qfunc$-function, as described in \Cref{eq: convex concave q}, which we restate here
\begin{equation}\label{eq:Qapp}
  \qfunc_\params(\state,\minaction,\maxaction)
  = N^1_{\params^1}(\state,\minaction) - N^2_{\params^2}(\state,\maxaction) + \minaction^\top M_{\params^3}(\state)\, \maxaction,
  \qquad \params = (\params^1,\params^2,\params^3) \in \Theta.
\end{equation}
Here, $\Theta$ is a convex set with $\|\params\|_2\leq B_\Theta~\forall~\params\in\Theta$. We take $M(\state)$ as a standard feedforward network given by
\begin{align}\label{eq: bilinear net in q arch}
    h^M_0=\state,\quad h^M_{\ell+1}=\varphi(W_\ell^M h_{\ell}^M + b_\ell^M), ~\ell=0,\cdots,L_M-1, \quad\textrm{vec}(M):=W_{\text{out}} h^M_{L_M},
\end{align}
where we take the hidden width as $p_M$. For convenience, we take the activation functions for $M$ to be the same as $N^i$, though $M$ can have any $1$-Lipschitz activation alternatives such as ReLU.
Our construction exploits two elementary properties of the softplus function $\varphi$. \Cref{eq: bilinear net in q arch} implies that $M$ has
\begin{align}
    |\params^3|=p_ML_M + (L_M-1)p_M^2 + p_Md_{\state} + p_Md_{\minaction}d_{\maxaction} \label{eq: param count bilinear}
\end{align}
parameters. Thus, $\qfunc_{\params}$ has $|\params^1|+|\params^2|+|\params^3|$ total parameters.
\begin{lemma}\label{lem:spl}
Let $\varphi(t) = \log(1+e^{t})$. Then
\begin{enumerate}
  \item[(i)] $\varphi$ is
    $1$-Lipschitz, strictly increasing and convex;
  \item[(ii)] for every $\bm a \in \mathbb{R}^{m}$, with $\varphi$ applied
    entrywise, $\;\|\varphi(\bm a)\|_2 \le \|\bm a\|_2 + \sqrt{m}\log 2$.
\end{enumerate}
\end{lemma}
 
\begin{proof}
(i) is immediate from $\varphi'(t) = e^t/(1+e^t) \in (0,1)$, and convexity
follows from $\varphi'' > 0$. For (ii), let $\bm{1} \in \mathbb{R}^m$ be the
all-ones vector and note $\varphi(0) = \log 2$. By (i) each coordinate obeys
$|\varphi(a_j) - \log 2| \le |a_j|$, so
$\|\varphi(\bm a) - (\log 2)\bm{1}\|_2^2 = \sum_j (\varphi(a_j)-\varphi(0))^2
\le \sum_j a_j^2 = \|\bm a\|_2^2$. The triangle inequality gives
$\|\varphi(\bm a)\|_2 \le \|\varphi(\bm a)-(\log2)\bm{1}\|_2 + \|(\log2)\bm{1}\|_2
\le \|\bm a\|_2 + \sqrt m \log 2$.
\end{proof}

The following assumption will be required to establish output and Lipschitz bounds, and is standard in neural network literature.

\begin{assumption}[Bounded inputs and weights]\label{ass:bdd}
 There exist finite positive constants $B_W, B_b, B_w, B_M$ and $D_{\statespace}$ such that:
 \begin{enumerate}
     \item For any weight matrix $W$ and bias vector $b$ in \Cref{eq:picnn}, the spectral norm $\|W\|_2\leq B_W$ and $\ell^2$-norm $\|b\|_2\leq B_b$.
     \item The output weights $w^i$ satisfy $\|w^i\|_2 \leq B_w$.
 \end{enumerate} 
 \end{assumption}
\Cref{ass:bdd} allows us to establish the Lipschitzness (with respect to the convex input) of a single PICNN $N^i$ as described in \Cref{eq:picnn}.
\subsection{Notation for Lipschitz and covering number bounds.}\label{subsection: appendix notation for picnn} Besides the constants in \Cref{ass:bdd}, we use several other notations throught the rest of this Section. For a scalar $t\in\real$, let the function $G(t,j)$ denote the sum $G(t,j)=\sum_{k=0}^j t^k$. From \cref{eq:picnn}, for $N^i$, we refer to the pre-activation layer of the convex path as $a^i_{\ell+1}$, i.e., $z^i_{\ell+1}=\varphi(a^i_{\ell+1})$, and we denote $D^i_{\ell+1}:=\mathrm{diag}\left(\varphi'(a^i_{\ell+1})\right)$ for all $\ell\in\{0,\dots,L-1\}$. Similarly, for all $\ell\in\{0,\dots,L-1\}$, we refer to the pre-activation layer of the state path as $\tilde{a}^i_{\ell+1}:=\tilde{W}^i_\ell h^i_\ell+\tilde{b}^i_\ell$, and define $\tilde{D}^i_{\ell+1}:=\textrm{diag}\left(\varphi'(\tilde{a}^i_{\ell+1})\right)$. For each layer $\ell$ of the PICNN $N^i$, we denote the vectors dependent only on $\x$ as
    \begin{align*}
c^{(z)}_\ell := \varphi\big(W^{(zh),i}_\ell h^i_\ell + b^{(z),i}_\ell\big) \in \mathbb{R}^m_{\ge 0}, \qquad c^{(y)}_\ell \;:=\; W^{(yh),i}_\ell h^i_\ell + b^{(y),i}_\ell \in \mathbb{R}^{d_{\bm y}},
\end{align*} 
with $\|c^{(z)}_\ell\|_\infty,\|c^{(y)}_\ell\|_\infty\leq B_WH+B_b+\log2=:C$, where $H$ is defined in \Cref{lemma: individual picnn lipschitz}, and $B_W,B_b$ defined in \Cref{ass:bdd}. We also denote $\rho:=B_W(C)$, and $C_Z=B_WH+B_b+\sqrt{m}\log2$. Finally, $B_{\Theta}$ is a positive constant such that $\|\params\|_2\leq B_\Theta~\forall~\params\in\Theta$.

\subsection{Lipschitz Bounds on Partially Input Convex-Conave Neural Network}\label{subsection: appendix lipschitz actions bounds F}

\begin{lemma}[Lipschitzness of PICNN]\label{lemma: individual picnn lipschitz}
    Under \Cref{assumption: actions,ass:bdd}, a single PICNN $N^i_\params(\x,\y$) as described in \Cref{eq:picnn} is $L_N$-Lipschitz in the input $\bm y$, where
    \begin{align*}
        L_N = B_w \cdot \rho \cdot G(\rho,L-1)
    \end{align*}
    where $\rho=B_W(B_WH+B_b+\log 2)$ from \Cref{subsection: appendix notation for picnn} and $H$ is defined as
    \begin{align*}
     H=\max\left(D_\statespace,(B_W)^L D_{\statespace} + (B_b + \sqrt{p}\log 2)\cdot G(B_W, L)\right),
    \end{align*}
     and $D_\statespace=\sup_{\state\in\statespace}\|\state\|_2<\infty$ from \Cref{assumption: actions}.
\end{lemma}
\begin{proof}
We will first show that $\|h^i_\ell\|_2 \leq H$ for all $\ell \in [L]$ and $\state\in\statespace$. From \Cref{eq:picnn}, we have $\|h^i_0\|_2 =\|\state\|\leq D_{\statespace}$, and
    \begin{align*}
        \|h^i_{\ell+1}\|_2 = \left\|\varphi\Big(\widetilde W^i_\ell\, h^i_\ell + \widetilde b^i_\ell\Big)\right\|_2 \leq B_W\|h^i_\ell\|_2 + B_b + \sqrt{p}\log 2.
    \end{align*}
    Unrolling this yields $H$.
    For each layer $\ell$, considering the vectors dependent only on $\x$,
 we have $\|c^{(z)}_\ell\|_\infty\leq \|W^{(zh),i}_\ell h^i_\ell + b^{(z),i}_\ell\|_2\leq B_WH+B_b+\log(2)$. Similarly, $\|c^{(y)}_\ell\|_\infty\leq B_WH+B_b$. 
For fixed $\bm x$, the vectors $c^{(z)}_\ell$ and $c^{(y)}_\ell$ can be treated as constants independent of $y$. Then, the forward pass reduces to:
\begin{align*}
    z^i_{\ell+1} = \varphi\Big(W^{(z),i}_\ell\,\mathrm{diag}\left(c^{(z)}_\ell\right)z^i_\ell + W^{(y),i}_\ell\,\mathrm{diag}\left(c^{(y)}_\ell\right)\bm y+\text{constant in }\bm y\Big).
\end{align*}Let $J_\ell := \|\partial z^i_\ell / \partial \bm y\|_2$. Differentiating through $\varphi$, since $\varphi' \in (0,1)$:
$$J_{\ell+1} \leq \|W^{(z),i}_\ell\|_2 \|c^{(z)}_\ell\|_\infty J_\ell + \|W^{(y),i}_\ell\|_2 \|c^{(y)}_\ell\|_\infty, \qquad J_0 = 0.$$
Note that both $\|W^{(z),i}_\ell\|_2 \|c^{(z)}_\ell\|_\infty$ and $\|W^{(y),i}_\ell\|_2 \|c^{(y)}_\ell\|_\infty$ are less than $\rho:=B_W(B_WH+B_b+\log 2)$. Then, unrolling $J_{\ell+1} \leq \rho J_\ell + \rho$ from $J_0=0$ yields $J_L\leq\sum_{k=1}^L\rho^k=\rho\cdot G(\rho,L-1)$.
Since $N^i(\x,\y) = (w^i)^\top z^i_L$ and $\|\partial N^i/\partial y\|_2 \le \|w^i\|_2 J_L \le B_w J_L$, the mean value theorem gives $|N^i(\x,\y)-N^i(\x,\y')|\le B_w J_L \|\y-\y'\| = L_N\|\y-\y'\|$.
\end{proof}
We are now ready to show that the convex-concave neural network function approximation class introduced in \Cref{section: convex concave fa} is Lipschitz in the players' actions. 
\begin{theorem}[Lipschitzness of $\qfunc\in\mathcal{F}$] \label{theorem: lipschitz action bound for F}
Under \Cref{assumption: actions,ass:bdd}, any function belonging to the convex-concave neural network function class introduced in \Cref{section: convex concave fa} with architecture specified in \cref{eq:picnn,eq: bilinear net in q arch} is $L_\mathcal{F}$-Lipschitz in the players' joint actions, with
\begin{align*}
    L_\mathcal{F}=2L_N + B_M\left(D_\actionspacemin+D_\actionspacemax\right),
\end{align*}
where $L_N$ is defined in \Cref{lemma: individual picnn lipschitz}, $D_\actionspacemin,D_\actionspacemax$ are defined in \Cref{assumption: actions}, and $\|M_\theta(s)\|_2 \leq B_M$ for all $s\in\statespace$. Note that a bound for $B_M$ will be established later on in \Cref{theorem: covering number bound}.
\end{theorem}
\begin{proof}
    For fixed $\state\in\statespace,~\params\in\Theta$, we have from \Cref{lemma: individual picnn lipschitz} that
    \begin{align*}
        &\big|Q_\theta(\state,\minaction,\maxaction)-Q_\theta(\state,\minaction',\maxaction')\big| \leq \underbrace{|N^1(\state,\minaction)-N^1(\state,\minaction')|}_{\leq L_N\|\minaction-\minaction'\|_2} + \underbrace{|N^2(\state,\maxaction)-N^2(\state,\maxaction')|}_{\leq L_N\|\maxaction-\maxaction'\|_2}\\
        & \qquad\qquad\qquad\qquad\qquad\qquad\qquad\qquad + |\minaction^\top M(\state)\maxaction - \minaction'^\top M(\state)\maxaction'|\\
        &\leq L_N\|\minaction-\minaction'\|_2 + L_N\|\maxaction-\maxaction'\|_2 +  |\minaction^\top M(\state)\maxaction - \minaction^\top M(\state)\maxaction' + \minaction^\top M(\state)\maxaction' - \minaction'^\top M(\state)\maxaction'|\\
        &\leq L_N\|\minaction-\minaction'\|_2 + L_N\|\maxaction-\maxaction'\|_2 + D_\actionspacemax B_M\|\minaction-\minaction'\|_2 + D_\actionspacemin B_M\|\maxaction-\maxaction'\|_2.
    \end{align*}
    This established the Lipschitz constant $L_\mathcal{F}$ of $\qfunc\in\mathcal{F}$ for the architecture specified in \cref{eq:picnn,eq: bilinear net in q arch}.
\end{proof}

\subsection{Upper Bound on the covering number of convex-concave $\qfunc$-function}\label{appendix: convering number of convex-concave q function}
\Cref{theorem: fqi bounds} depends on the covering number of the convex-concave neural network function class $\mathcal{F}$ introduced in \Cref{section: convex concave fa} in the sup norm, $N(\eta):=N(\mathcal{F},\eta,\|\cdot\|_\infty)$ for the case when $\qfunc_\params\in\mathcal{F}$ is constructed according to \cref{eq:picnn,eq: bilinear net in q arch}. In this section, we provide a bound on this covering number.
\begin{theorem}[Covering number bound]\label{theorem: covering number bound}
    Under \Cref{assumption: actions,ass:bdd}, for convex-concave networks $\qfunc_{\params}\in\mathcal{F}$ constructed according to \cref{eq:picnn,eq: bilinear net in q arch} with $\|\params\|_2\leq B_\Theta~\forall~\params\in\Theta$, we have
    \begin{align*}
        N(\mathcal{F},\eta,\|\cdot\|_\infty)\leq \left(\frac{2B_{\Theta}\Gamma}{\eta}+1\right)^{|\params^1|+|\params^2|+|\params^3|}
    \end{align*}
    with $\Gamma=\sqrt{2L^{\textnormal{PICNN}}_{\params}+L^M_\params}$, where 
    \begin{align*}
&L^{\textnormal{PICNN}}_{\params} =\left(C^2Z^2+C^2D_{\y}^2+H^2+\left(B_W^2Z^2+B_W^2D_{\y}^2\right)\left(H^2+1\right)+1\right) B_{\delta}+\left(H^2+1\right)B_{\tilde\delta}+Z^2,\\
&L^M_{\params}=D_{\actionspacemin}^2D_{\actionspacemax}^2\left(B_W\right)^{2L_M}\left(\left(H^M\right)^2+1\right)G\left(\frac{1}{B_W^2},L_M-1\right)+D^2_{\actionspacemin}D^2_{\actionspacemax}{H^M}^2\\
&B_\delta=B_w^2G\left(\rho^2, L-1\right),\\ &B_{\tilde\delta}=\sum_{\ell=0}^{L-1} \Big(\left(B_W^2(Z+ D_{\y})+B_W\right)B_w \rho^{L-\ell-1} G\left(\nicefrac{B_W}{\rho},L-1-\ell\right)\Big)^2\\
&Z = \left(D_{\y}\rho+C_Z\right)G(\rho,L-1), \quad D_{\y}=\max{D_{\actionspacemin},D_{\actionspacemax}},
    \end{align*}
    $C, C_Z, \rho$ are defined in \Cref{subsection: appendix notation for picnn}, $H$ is the constant defined in \Cref{lemma: individual picnn lipschitz}, $H^M$ is defined in \cref{eq: bound on HM} and $|\params^1|,|\params^2|$, and $|\params^3|$ are defined in \cref{eq: param count picnn,eq: param count bilinear} 
    
    Further, for $\sup_{\state}\|M(\state)\|_2\leq B_M$ used in \Cref{theorem: lipschitz action bound for F}, we have $B_M=B_WH^M$.
\end{theorem}
\begin{proof} 
To bound the covering number, we need a bound on the Lipschitz constant of $\qfunc_\params\in\mathcal{F}$ with respect to the network parameters $\params$, which can in turn be found by bounding $\sup_{\state\minaction,\maxaction}\|\nabla_\params\qfunc_\params\|_2$ uniformly over all $\theta\in\Theta$. Let $\delta^i_\ell := \nicefrac{\partial N^i}{\partial z^i_\ell}$, a bound on $\|\delta^i_\ell\|$ will help us bound the gradients with respect to all network parameters. By definition, $\delta^i_L = w^i$, and for $\ell=0,\dots,L-1$, by the chain rule
\begin{align*}
    \delta^i_\ell = \left(\frac{\partial z^i_{\ell+1}}{\partial z^i_\ell}\right)^\top& \delta^i_{\ell+1}=\text{diag}(c^{(z)}_\ell)\left(W^{(z),i}_\ell\right)^\top\text{diag}(\varphi'(a^i_{\ell+1}))\delta^i_{\ell+1}\\
    \implies& \|\delta^i_\ell\|_2\leq \|c^{(z)}_\ell\|_\infty B_W \|\delta^i_{\ell+1}\|_2=CB_W\|\delta^i_{\ell+1}\|_2\\
\end{align*}
Starting from $\|\delta_L^i\|_2=\|w^i\|_2\leq B_w$, this recursion implies $\|\delta^i_\ell\|_2\leq B_w \rho^{L-\ell}$. Let us define
\begin{align*}
\sum_{\ell=0}^{L-1} B_w^2\rho^{2(L-\ell-1)}=B_w^2G\left(\rho^2, L-1\right):=B_{\delta}
\end{align*}

We will also use a bound on $\|z_\ell^i\|_2$. We have $z_\ell^0=\bm 0$, and from \Cref{eq:picnn}, $\|z_{\ell+1}^i\|_2\leq \rho\|z_\ell^i\|_2 + D_{\bm y}\rho + C_Z\implies \|z^i_\ell\|_2\leq \left(D_{\y}\rho+C_Z\right)G(\rho,\ell-1)$ implying that for all $\ell\in\{0,\dots,L\}$, $\|z^i_\ell\|_2\leq \|z^i_L\|_2\leq \left(D_{\y}\rho+C_Z\right)G(\rho,L-1)=:Z$. We will use this result to bound the gradients with respect to all parameters. Starting with the parameters on the convex path, we have for $r^{(z),i}_\ell:=W^{(zh),i}_\ell h^i_\ell + b^{(z),i}_\ell$:
\begin{align}\label{eq: deriv convex picnn bound}
\begin{aligned}
     &\nabla_{W^{(z),i}_\ell}N^i = D^i_{\ell+1}\delta^i_{\ell+1}(z^i_\ell\odot c^{(z)}_\ell)^\top\implies\|\nabla_{W^{(z),i}_\ell}N^i\|_F\leq \|\delta^i_{\ell+1}\|_2\cdot C Z,\\
    &\nabla_{W^{(y),i}_\ell}N^i = D^i_{\ell+1}\delta^i_{\ell+1}(\y\odot c^{(y)}_\ell)^\top\implies \|\nabla_{W^{(y),i}_\ell}N^i\|_F\leq \|\delta^i_{\ell+1}\|_2\cdot C D_{\y},\\
    &\nabla_{W^{(h),i}_\ell}N^i = D^i_{\ell+1}\delta^i_{\ell+1}(h^i_\ell)^\top\implies \|\nabla_{W^{(h),i}_\ell}N^i\|_F\leq \|\delta^i_{\ell+1}\|_2\cdot H,\\
    &\nabla_{W^{(zh),i}_\ell}N^i = \left(\varphi'\left(r^{(z),i}_\ell\right)\odot z^i_{\ell}\odot\left(\left(W^{(z),i}_{\ell}\right)\right)^\top D^i_{\ell+1}\delta^i_{\ell+1}\right)(h^i_\ell)^\top\\
    &\qquad\qquad\qquad\implies\|\nabla_{W^{(zh),i}_\ell}N^i\|_F\leq \|\delta^i_{\ell+1}\|_2\cdot B_WZH,\\
    &\nabla_{W^{(yh),i}_\ell}N^i = \left(\left(\left(W^{(y),i}_{\ell}\right)^\top D^i_{\ell+1}\delta^i_{\ell+1}\right)\odot  \y\right)(h^i_\ell)^\top\\
    &\qquad\qquad\qquad\implies\|\nabla_{W^{(yh),i}_\ell}N^i\|_F\leq \|\delta^i_{\ell+1}\|_2\cdot B_WD_{\y}H\\
    &\nabla_{b^{(z),i}_\ell}N^i = \varphi'\left(r^{(z),i}_\ell\right)\odot z^i_{\ell}\odot\left(\left(W^{(z),i}_{\ell}\right)^\top D^i_{\ell+1}\delta^i_{\ell+1}\right)\implies\|\nabla_{b^{(z),i}_\ell}N^i\|_2\leq \|\delta^i_{\ell+1}\|_2\cdot B_WZ,\\
    &\nabla_{b^{(y),i}_\ell}N^i = \left(\left(W^{(y),i}_{\ell}\right)^\top D^i_{\ell+1}\delta^i_{\ell+1}\right)\odot  \y\implies\|\nabla_{b^{(y),i}_\ell}N^i\|_2\leq \|\delta^i_{\ell+1}\|_2\cdot B_WD_{\y}, \text{ and}\\
    &\nabla_{b^i_\ell}N^i = D_{\ell+1}^i\delta^i_{\ell+1}\implies\|\nabla_{b^i_\ell}N^i\|_2\leq \|\delta^i_{\ell+1}\|_2,
\end{aligned}
\end{align}

Similarly, for parameters on the state path, we will need a bound on $\tilde{\delta}^i_\ell:=\nabla_{h^i_\ell}N^i$. We have $\tilde{\delta}^i_\ell = \left(\nicefrac{\partial h^i_{\ell+1}}{\partial h^i_\ell}\right)^\top \tilde{\delta}^i_{\ell+1} + \left(\nicefrac{\partial a^i_{\ell+1}}{\partial h^i_\ell}\right)^\top \left(\nicefrac{\partial N^i}{\partial a^i_{\ell+1}}\right)$.
From \Cref{subsection: appendix notation for picnn}, $\tilde{a}^i_{\ell+1}=\tilde{W}^i_\ell h^i_\ell+\tilde{b}^i_\ell$, and $\tilde{D}^i_{\ell+1}=\textrm{diag}\left(\varphi'(\tilde{a}^i_{\ell+1})\right)$. Then the first derivative above  $\nicefrac{\partial h^i_{\ell+1}}{\partial h^i_\ell}=\tilde{D}^i_{\ell+1}\tilde{W}^i_\ell$, and the second derivative $\nicefrac{\partial a^i_{\ell+1}}{\partial h^i_\ell}=W^{(z),i}_\ell\textrm{diag}\left(z^i_\ell\odot\varphi'\left(r^{(z),i}_\ell\right)\right)W^{(zh),i}_\ell + W^{(y),i}_\ell\textrm{diag}\left(\y\right)W^{(yh),i}_\ell + W^{(h),i}_\ell$, yielding
\begin{align*}
    \tilde{\delta}^i_\ell &=\left(W^{(z),i}_\ell\textrm{diag}\left(z^i_\ell\odot\varphi'\left(r^{(z),i}_\ell\right)\right)W^{(zh),i}_\ell + W^{(y),i}_\ell\textrm{diag}\left(\y\right)W^{(yh),i}_\ell + W^{(h),i}_\ell\right)^\top D_{\ell+1}\delta^i_{\ell+1}\\
    &+\left(\tilde{W}^i_\ell\right)^\top \tilde{D}^i_{\ell+1}\tilde{\delta}^i_{\ell+1}\\
    \implies& \|\tilde{\delta}^i_\ell\|_2 \leq B_W\|\tilde{\delta}^i_{\ell+1}\|_2  + \left(B_W^2(Z+D_{\y}) + B_W\right)B_w\rho^{L-\ell-1} 
\end{align*}
Unrolling this recursion with $\|\tilde\delta^i_{L-1}\|_2\leq \left(B_W^2(Z+ D_{\y})+B_W\right)B_w$ yields $\|\tilde\delta^i_{\ell}\|_2\leq \left(B_W^2(Z+ D_{\y})+B_W\right)B_w \rho^{L-\ell-1} G\left(\nicefrac{B_W}{\rho},L-1-\ell\right)$. We define
\begin{align*}
    \sum_{\ell=0}^{L-1} \left(\left(B_W^2(Z+ D_{\y})+B_W\right)B_w \rho^{L-\ell-1} G\left(\nicefrac{B_W}{\rho},L-1-\ell\right)\right)^2=:B_{\tilde\delta}
\end{align*}

Using this, we have
\begin{align*}
    &\nabla_{\tilde{W}^i_\ell}N^i = (\tilde{D}^i_{\ell+1}\tilde{\delta}^i_{\ell+1})(h^i_\ell)^\top\implies\|\nabla_{\tilde{W}^i_\ell}N^i\|_F\leq \|\tilde{\delta}^i_{\ell+1}\|_2\cdot H\\
    &\nabla_{\tilde{b}^i_\ell}N^i = (\tilde{D}^i_{\ell+1}\tilde{\delta}^i_{\ell+1})\implies\|\nabla_{\tilde{W}^i_\ell}N^i\|_F\leq \|\tilde{\delta}^i_{\ell+1}\|_2
\end{align*}

Now for each PICNN $N^i$, we have
\begin{align}\label{eq: deriv state picnn bound}
   \begin{aligned}
        \|\nabla_{\params^i}N^i\|_2^2 &= \sum_{\ell=0}^{L-1}\Big(\|\nabla_{W^{(z),i}_\ell}N^i\|_F^2+\|\nabla_{W^{(y),i}_\ell}N^i\|_F^2+\|\nabla_{W^{(h),i}_\ell}N^i\|_F^2+\|\nabla_{W^{(zh),i}_\ell}N^i\|_F^2\\
    &+\|\nabla_{W^{(yh),i}_\ell}N^i\|_F^2+\|\nabla_{b^{(z),i}_\ell}N^i\|_2^2+\|\nabla_{b^{(y),i}_\ell}N^i\|_2^2+\|\nabla_{b^{i}_\ell}N^i\|_2^2\\
    &+\|\nabla_{\tilde{W}^i_\ell}N^i\|_F^2 + \|\nabla_{\tilde{b}^i_\ell}N^i\|_2^2+\|\nabla_{w^i}N^i\|_2^2\Big).
   \end{aligned}
\end{align}
Using the bounds for each set of derivatives from \Cref{eq: deriv convex picnn bound,eq: deriv state picnn bound}, we have
\begin{align}
   \nonumber &\sup_{\params^i}\|\nabla_{\params^i}N^i\|_2^2 \leq \left(C^2Z^2+C^2D_{\y}^2+H^2+\left(B_W^2Z^2+B_W^2D_{\y}^2\right)\left(H^2+1\right)+1\right) \sum_{\ell=0}^{L-1}\|\delta^i_{\ell+1}\|_2^2\\
   \nonumber &\quad\qquad\qquad\qquad+\left(H^2+1\right)\sum_{\ell=0}^{L-1}\|\tilde{\delta}^i_{\ell+1}\|_2^2+Z^2\\
    &\leq \underbrace{\left(C^2Z^2+C^2D_{\y}^2+H^2+\left(B_W^2Z^2+B_W^2D_{\y}^2\right)\left(H^2+1\right)+1\right) B_{\delta}+\left(H^2+1\right)B_{\tilde\delta}+Z^2}_{=:L^{\text{PICNN}}_{\params}}\label{eq: param lip picnn}
\end{align}
For the standard feedforward network $M_\params^3(\state)$ defined in \cref{eq: bilinear net in q arch}, we will first bound $\|h^M_\ell\|_2$.
We have $h^M_{k+1}=\varphi(W_kh_{k}^M + b_k) \implies  \|h^M_{\ell+1}\|_2\leq B_W \|h^M_\ell\|_2 +B_b+\sqrt{p_M}\log 2$. Unrolling this recursion with $\|h^M_0\|_2=\|\state\|_2\leq D_{\state}$ yields $\|h^M_\ell\|_2\leq D_{\statespace}B_W^\ell + (B_b+\sqrt{p_M}\log2)\cdot G(B_W,\ell-1)$ for $\ell>0$. This gives that
\begin{align}
    \|h^M_\ell\|\leq \max\left(D_{\statespace}, D_{\statespace}B_W^{L_M}+ (B_b+\sqrt{p_M}\log2)\cdot G(B_W,L_M-1)\right)=:H^M \label{eq: bound on HM}
\end{align}
Note that \Cref{eq: bound on HM} also gives $\sup_{\state}\|M\|_2\leq\| W_{\text{out}}h^M_{L_M}\|_F\leq B_W H^M=:B_M$
Let $w=\text{vec}(\minaction\maxaction^\top)$, then note that $\|w\|_2\leq D_{\actionspacemin}D_{\actionspacemax}$. Let $f=\minaction^\top M\maxaction$, then   $f= \minaction^\top M \maxaction = \left(\text{vec}(\minaction\maxaction^\top)\right)^\top\text{vec}(M) = w^\top W_{\text{out}}h^M_{L_M}$.  We have
\begin{align*}
    \nabla_{W_{\text{out}}}f=w(h^M_{L_M})^\top \implies \|\nabla_{W_{\text{out}}}f\|_F\leq D_{\actionspacemin}D_{\actionspacemax}\cdot \|h^M_{L_M}\|
\end{align*}
Let $\delta^M_\ell:=\nabla_{a^M_\ell}f$, then $\delta^M_\ell=\text{diag}\left(\varphi'(a^M_\ell)\right)(W^M_\ell)^\top\delta^M_{\ell+1}$. Unrolling this recursion with $\|\delta^M_{L_M}\|_2\leq B_WD_{\actionspacemin}D_{\actionspacemax}$ yields $\|\delta^M_\ell\|_2 \leq B_W^{L_M-\ell+1}D_{\actionspacemin}D_{\actionspacemax}$. Now we have
\begin{align*}
    &\nabla_{W^M_\ell}f = \delta^M_{\ell+1}(h^M_\ell)^\top\implies\|\nabla_{W^M_\ell}f\|_F \leq B_W^{L_M-\ell}D_{\actionspacemin}D_{\actionspacemax}\|h^M_\ell\|_2\\
    &\nabla_{b^M_\ell}f = \delta^M_{\ell+1}\implies\|\nabla_{W^M_\ell}f\|_F \leq B_W^{L_M-\ell}D_{\actionspacemin}D_{\actionspacemax}
\end{align*}
Thus
\begin{align}
    \nonumber&\|\nabla_{\params^3}\minaction^\top M_{\params^3}(\state)\maxaction\|_2^2=\|\nabla_{W_{\text{out}}} f\|_F^2 + \sum_{\ell=0}^{L_M-1}\|\nabla_{W^M_\ell}f\|_F^2 + \|\nabla_{b^M_\ell}f\|_2^2\\
    \nonumber\implies\sup_{\params^3}&\|\nabla_{\params^3}\minaction^\top M_{\params^3}(\state)\maxaction\|_2^2\\
    &\leq \underbrace{D_{\actionspacemin}^2D_{\actionspacemax}^2\left(B_W\right)^{2L_M}\left(\left(H^M\right)^2+1\right)G\left(\frac{1}{B_W^2},L_M-1\right)+D^2_{\actionspacemin}D^2_{\actionspacemax}{H^M}^2}_{:=L^M_\params}.\label{eq: param lip bilinear}
\end{align}
Thus, using \cref{eq: param lip picnn,eq: param lip bilinear}, any $\qfunc_\params\in\mathcal{F}$ constructed according to the architecture specified in \cref{eq:picnn,eq: bilinear net in q arch} satisfies
\begin{align*}
    \|\sup_{\params}\nabla_{\params}\qfunc_\params\|_2 \leq   \sqrt{2L^{\text{PICNN}}_{\params} + L^M_{\params}}=\Gamma.
\end{align*}
Thus, the $\qfunc_\theta$ network in the function class $\mathcal{F}$ with the architecture specified in \cref{eq:picnn,eq: bilinear net in q arch} is $\Gamma$-Lipschitz in the parameters $\params$. Correspondingly, using the facts that the required covering number $N(\mathcal{F}, \eta, \|\cdot\|_\infty)\leq N(\Theta, \nicefrac{\eta}{\Gamma}, \|\cdot\|_2)$, and that for positive $B_\Theta$ such that $\|\params\|_2\leq B_\Theta~\forall~\params\in\Theta$,  $N(\Theta, \tilde{\eta}, \|\cdot\|_2) \leq \left(\nicefrac{2B_\Theta}{\tilde{\eta}}+1\right)^{|\params|}$ when $\Theta$ is any convex set within the $B_\Theta$-bounded $\ell^2$ ball in $R^{|\params|}$ (c.f., Chapter 4 of \citet{Vershynin_2026}), we have
$N(\mathcal{F},\eta,\|\cdot\|_\infty)\leq \left(\nicefrac{2B_{\Theta}\Gamma}{\eta}+1\right)^{|\params^1|+|\params^2|+|\params^3|}$, with $|\params^1|,|\params^2|$, and $|\params^3|$ from \cref{eq: param count picnn,eq: param count bilinear}. 
\end{proof}

\section{Experimental Details}\label{appendix: experimental details}
In this section, we provide additional implementation details for the experiments in \Cref{section: experiments}.

\paragraph{Convex-Concave $\qfunc$-function.} We used the architecture described in \cref{eq:picnn,eq: bilinear net in q arch}. Each partially convex input network has $2$ hidden layers with hidden depth $32$ for both state and convex paths. We used softplus activation for the convex path, and ReLU activation for the state path. For the bilinear network $M$, we used 1 hidden layer with width $32$, and ReLU activation.

\paragraph{Hardware.} All experiments were carried out on a desktop with an AMD Ryzen 9 7950X 16-Core Processor.

\end{document}